\documentclass[10pt]{article}

\usepackage[utf8]{inputenc}
\usepackage[left=2.5cm,right=2.5cm,top=2.5cm,bottom=2.5cm]{geometry}
\usepackage[numbers]{natbib}
\usepackage{amsmath}
\usepackage{amssymb}
\usepackage{amsfonts}
\usepackage{amsthm}
\usepackage{textcomp}
\usepackage{xcolor}
\usepackage{graphicx}
\usepackage{booktabs}
\usepackage{authblk}
\usepackage{hyperref}
\usepackage{caption}
\newtheorem{theorem}{Theorem}
\newtheorem{lemma}{Lemma}
\newtheorem{assumption}{Assumption}
\newtheorem{definition}{Definition}
\newtheorem{corollary}{Corollary}
\newtheorem{remark}{Remark}
\newtheorem{problem}{Problem}
\usepackage{algorithm}
\usepackage{algpseudocode}
\algrenewcommand\algorithmicrequire{\textbf{Input:}}
\algrenewcommand\algorithmicensure{\textbf{Output:}}
\DeclareMathOperator{\trace}{tr}

\date{}

\title{\textbf{Information-Weighted Distributed Kalman Filter for Multi-Target Tracking over Switching Sensor Networks}}

\author[1]{Tirthankar Chakraborty \thanks{Email: tirthankar@aero.iitb.ac.in}}
\author[1]{Arnab Maity}

\affil[1]{\small{Department of Aerospace Engineering, Indian Institute of Technology Bombay, Mumbai 400076, India}}

\begin{document}

\maketitle



\textbf{\textit{Abstract}} --- Target tracking over wide areas using a network of sensor agents is a challenging problem when agents have limited sensing coverage and communication topologies vary over time. Existing distributed filtering approaches commonly assume that each agent, or at least one of its neighbors, can directly observe the target, and further rely on fixed network topologies, global network parameters, or multiple consensus rounds per time step. These requirements limit their applicability in wide-area network deployments. In this paper, we propose an Information-Weighted Distributed Kalman Filter (IWDKF) that relaxes these assumptions through a weighted-consensus mechanism applicable to arbitrary switching topologies satisfying joint connectivity. The proposed approach requires no prior knowledge of global network properties and limits communication overhead to a single consensus exchange per time step, while enabling information fusion even when some agents lack direct target measurements. The weighted-consensus covariance matrix is shown to remain positive definite, and the resulting consensus update is proven to maintain estimate consistency. It is also established that the proposed weighting strategy preserves the asymptotic agreement property of the underlying consensus protocol under joint connectivity. Numerical simulations of a multi-target tracking scenario involving a network of sensor agents demonstrate the improved estimation accuracy of the proposed IWDKF under both fixed and switching network topologies. The performance of the proposed filter is close to that of the centralized Kalman filter, and is compared with three existing distributed filters.\\

\textbf{\textit{Keywords}} --- Target tracking, sensor networks, distributed state estimation, information fusion, switching communication topologies, consensus protocols



\section{Introduction} \label{sec:intro}

Distributed sensor networks consist of geographically separated yet interconnected sensor agents that collaborate to sense and monitor their environment. They provide improved robustness and fault tolerance compared to isolated agents and find applications in both civilian and military domains, including cooperative target tracking \cite{hu2021distributed}, structural health monitoring \cite{li2022unmanned}, disaster response \cite{hu2022fault}, and aerial reconnaissance \cite{di2015potential}. State estimation is a fundamental component of many of these applications, enabling agents to infer the state of dynamic targets or vehicles from noisy sensor measurements. More than six decades after its introduction, the Kalman filter \cite{kalman1960new} has remained a standard solution for state estimation in noisy environments \cite{anderson2005optimal}. A \textit{centralized} approach, in which all agents send their data to a central fusion center, becomes impractical as network size increases due to bandwidth and computational constraints and introduces a single point of failure. By contrast, a distributed approach \cite{jin2022distributed, pedroso2022discrete} shares the computational load across all agents, thereby improving resilience and removing dependence on a central node.


\subsection{Related Literature}

A decentralized Kalman filter that requires all-to-all communication was proposed in \cite{rao1993fully}, which improves robustness but not scalability. This decentralized architecture has also been extended to networks with node-specific states, in which covariance intersection (CI) is used to account for unknown correlations among local estimators \cite{saini2019decentralized, saini2022adaptive}. Subsequently, consensus-based distributed estimation methods that rely solely on local inter-neighbor communication were developed, enabling such methods to scale effectively to large networks. The first scalable distributed Kalman filter was introduced in \cite{olfati2005distributed}, using embedded \textit{average consensus} filters, though it was later shown to exhibit suboptimal performance \cite{olfati2007distributed}.

The Kalman consensus filter (KCF) \cite{olfati2009kalman} is a widely recognized distributed filtering approach that augments the update step with an average consensus term. However, its performance degrades when some sensors are \textit{oblivious} (i.e., lack target measurements), a common issue in wide-area networks with limited sensing ranges. This issue was addressed in \cite{kamal2011generalized} using a generalized Kalman consensus filter (GKCF) that employs an embedded weighted average consensus algorithm. This was further examined in \cite{kamal2013information}, where the information-weighted consensus filter (ICF) was proposed and was shown to converge asymptotically to the centralized solution, provided that the agents perform infinitely many consensus iterations with their neighbors prior to each local update. A related information-driven filter was presented in \cite{wang2014information} that requires neither global network parameters nor a separate consensus loop. However, it defines the weighted-consensus term as a uniform average over the inclusive neighborhood and does not provide a proof of convergence. Rather than fusing neighboring estimates, the approach in \cite{di2015distributed} selects the most accurate local estimate and propagates it across the network, thereby excluding the estimates of poorly informed agents.

A dynamic active-average consensus algorithm was introduced in \cite{peterson2019active}, distinguishing \textit{active} agents (with target measurements) from \textit{passive} agents (without target measurements), with the objective of driving all agents to converge to the average of the active agents' measurements \cite{chung2020dynamic}. To overcome the loss of optimality caused by oblivious agents, the approximate distributed Kalman filter (ADKF) was proposed in \cite{battilotti2020asymptotically} for continuous-time systems, demonstrating asymptotic convergence to the centralized solution. For discrete-time systems, a near-optimal solution that requires a sufficiently large number of consensus steps was presented in \cite{battilotti2021stability}. Additional consensus-based filters were proposed in \cite{millan2012distributed,matei2012consensus,battistelli2014consensus}.

In contrast to consensus-based approaches that emphasize agreement among agents, diffusion-based filters directly fuse local and neighboring information to improve individual estimates. A diffusion Kalman filter (DKF) was first introduced in \cite{cattivelli2010diffusion} under the assumption that each agent, or at least one of its neighbors, can sense the target. A DKF based on CI (CI-DKF) was proposed in \cite{hu2011diffusion}, which relaxes this assumption. Subsequently, a distributed hybrid information fusion (DHIF) algorithm was developed in \cite{wang2017convergence}, which achieves consistent estimation under mild conditions without requiring global parameters. An alternative approach was proposed in \cite{battistelli2014kullback}, which focuses on reaching consensus on the probability density functions associated with local estimates. The stability conditions for this approach were subsequently analyzed in \cite{battistelli2014parallel} under a fixed network topology and a linear Gaussian system model.


\subsection{Motivation}

Most existing approaches in this domain implicitly or explicitly assume that the target can be sensed by each sensor agent or at least one of its neighbors \cite{olfati2009kalman, cattivelli2010diffusion}. In wide-area sensor networks, it is common for a target to lie outside the sensing range of a cluster of agents, leaving those agents without direct target measurements. This limitation is further compounded for agents with limited fields of view (FOV), such as unmanned aerial vehicles (UAVs) carrying gimbaled cameras. Hence, this assumption does not hold in general for sensor networks. Another assumption is that each agent requires knowledge of certain global network parameters, such as the maximum degree of the network graph or the number of agents in the network \cite{kamal2011generalized, kamal2013information}. In practice, however, these parameters may change dynamically as agents join the network or fail and become inactive.

Existing methods also assume that the communication network topology is time-invariant \cite{olfati2007distributed, olfati2009kalman, kamal2011generalized, kamal2013information}. In practice, switching communication topologies can arise for a variety of reasons. These include intentional design choices, such as scheduled communication protocols, and unintentional factors, such as communication link failures and variations in the relative positions of mobile agents. Many approaches additionally require multiple consensus iterations prior to each local update \cite{olfati2009kalman, kamal2011generalized, kamal2013information}. This necessitates multiple rounds of communication among neighboring agents before each local update, which can significantly increase communication overhead in energy-constrained systems such as battery-powered UAVs.

Collectively, these assumptions limit the applicability and robustness of existing distributed filtering approaches in practical wide-area sensor network deployments.


\subsection{Contribution}

To address the aforementioned limitations, this paper develops a distributed filtering framework for target tracking using a network of agents with limited sensing capabilities and switching communication topologies. The proposed framework neither requires knowledge of global network parameters nor does it require that each agent or any of its neighbors detect the target. Furthermore, each agent communicates with its neighbors only once per time step. The main contributions of this work are summarized as follows:
\begin{enumerate}
    \item Development of a new Information-Weighted Distributed Kalman Filter (IWDKF) for multi-target tracking.
    
    \item Proof of positive definiteness of the covariance matrix associated with the weighted-consensus estimate.
    
    \item Convergence analysis establishing that the proposed weighting strategy achieves consensus at a geometric rate under jointly connected switching graphs.
    
    \item Proof that the weighted-consensus update preserves the consistency of locally consistent estimates.
    
    \item Performance evaluation of the proposed IWDKF, in a multi-target tracking scenario with a network of agents, through comparisons with the centralized Kalman filter and three existing distributed filters, namely, KCF, GKCF, and ICF, using numerical simulations.
\end{enumerate}

A preliminary version of this work, restricted to single-target tracking under topologies connected at every time instant and without convergence guarantees, was presented in our conference paper \cite{chakraborty2026dynamic}. The present paper generalizes that formulation to multi-target tracking under the strictly weaker assumption of jointly connected switching topologies, and additionally establishes the convergence and consistency guarantees not addressed in the conference version. The numerical evaluation is also extended to validate these theoretical contributions. Thus, the present paper constitutes substantial additional work compared to our conference paper.


The remainder of this paper is organized as follows. Section~\ref{sec:prelim} presents the necessary mathematical preliminaries. Section~\ref{sec:problem} formulates the problem. Section~\ref{sec:background} presents background on centralized and distributed estimation frameworks. Section~\ref{sec:proposed} introduces the proposed distributed filter. Section~\ref{sec:convergence} establishes the convergence and consistency properties of the proposed weighted-consensus mechanism. Section~\ref{sec:sim} demonstrates the effectiveness of the proposed method in a multi-target tracking scenario through numerical simulations. Finally, Section~\ref{sec:conclusion} concludes the paper.


\section{Preliminaries} \label{sec:prelim}

This section introduces the notation, graph-theoretic concepts, and the average consensus algorithm used throughout the remainder of the paper.


\subsection{Notation}

Vectors are denoted by bold lowercase letters and matrices by bold uppercase letters. $\mathbf{0}$ and $\mathbf{1}$ represent column vectors of all zeros and all ones, respectively, with appropriate dimensions. The identity matrix of order $m$ is denoted by $\mathbf{I}_m$. The Kronecker product is represented by the symbol ``$\otimes$''. The transpose and expectation operators are written as $(\cdot)^\top$ and $\mathbb{E}[\cdot]$, respectively. A Gaussian distribution with mean $\boldsymbol{\mu}$ and covariance $\boldsymbol{\Sigma}$ is denoted by $\mathcal{N}\big(\boldsymbol{\mu},\boldsymbol{\Sigma}\big)$. The notation $a \in \{ 1,\ldots,b\}$ is abbreviated as $\{a\}_1^b$. The sets of real numbers and positive integers are denoted by $\mathbb{R}$ and $\mathbb{Z}^+$, respectively. The sets of $n \times n$ symmetric positive semidefinite and positive definite matrices are represented by $\mathbb{S}^n_{+}$ and $\mathbb{S}^n_{++}$, respectively. $\mathbf{E}_1 \prec \mathbf{E}_2$ implies that $(\mathbf{E}_2 - \mathbf{E}_1) \in \mathbb{S}^n_{++}$, while $\mathbf{E}_1 \preceq \mathbf{E}_2$ signifies that $(\mathbf{E}_2 - \mathbf{E}_1) \in \mathbb{S}^n_{+}$, for any two symmetric matrices $\mathbf{E}_1, \mathbf{E}_2 \in \mathbb{S}^n$. Also, $\mathbf{E}_1 < \mathbf{E}_2$ implies that $(\mathbf{E}_2 - \mathbf{E}_1)$ is a positive matrix (entrywise positive), while $\mathbf{E}_1 \leq \mathbf{E}_2$ signifies that $(\mathbf{E}_2 - \mathbf{E}_1)$ is a nonnegative matrix (entrywise nonnegative). The trace of $\mathbf{E}$ is given by $\trace(\mathbf{E})$. The vectorization operator is denoted by $\operatorname{vec}(\cdot)$, which maps $\mathbf{E} \in \mathbb{R}^{m \times n}$ to a vector in $\mathbb{R}^{mn}$ formed by stacking its columns. The oscillation of a vector $\mathbf{d}$ is defined as $\operatorname{osc}(\mathbf{d}) \triangleq \max_i d_i-\min_i d_i$. The largest integer less than or equal to a real number $n$ is written as $\lfloor n \rfloor$.


\subsection{Graph Theory}

The communication topology of a wide-area sensor network with $N$ agents is modeled by the undirected graph $\mathcal{G}(k) = (\mathcal{V},\mathcal{E}(k))$, where $\mathcal{V} = \{1,2,\ldots,N\}$ denotes the set of sensor agents and $\mathcal{E}(k)\subseteq \mathcal{V}\times\mathcal{V}$ represents the set of communication links. The existence of an edge $(i,j)$ in $\mathcal{G}(k)$ indicates that agents $i$ and $j$ can exchange information. Since the graph is undirected, the edge relation is symmetric: $(i,j) \in \mathcal{E}(k) \Leftrightarrow (j,i) \in \mathcal{E}(k)$. The degree of a vertex in $\mathcal{G}(k)$ is the number of edges incident to that vertex, and $\Delta(k)$ denotes the maximum degree of $\mathcal{G}(k)$.

Two distinct agents $i$ and $j$ are \textit{neighbors} if $(i,j)\in \mathcal{E}(k)$. The neighborhood of agent $i$ is defined as $\mathcal{N}_i(k) \triangleq \{ j \in \mathcal{V} \mid (i,j) \in \mathcal{E}(k), j \neq i \}$, and the \textit{inclusive neighborhood} of agent $i$ is defined as $\mathcal{J}_i(k) \triangleq \mathcal{N}_i(k) \cup \{i\}$. A \textit{path} in a graph is a sequence of distinct edges connecting a sequence of distinct vertices. A graph is \textit{connected} if every pair of vertices is joined by at least one path. For a sequence of switching graphs $\{\mathcal{G}(k)\}$, the \textit{union graph} over the interval $[k_1,k_2]$ is defined as $\bigcup_{k=k_1}^{k_2}\mathcal{G}(k) = \left(\mathcal{V}, \bigcup_{k=k_1}^{k_2}\mathcal{E}(k)\right)$. The graph sequence is said to be \textit{jointly connected} over the interval $[k_1,k_2]$, if its union graph is connected. 

A matrix is said to be \textit{row-stochastic} if it has nonnegative entries and each row sums to unity. The adjacency matrix associated with $\mathcal{G}(k)$ is denoted by $\mathbf{A}(k)=[a_{ij}(k)]$, where $a_{ij}(k)=1$ if $(i,j)\in\mathcal{E}(k)$, and $a_{ij}(k)=0$ otherwise.


\subsection{Average Consensus}

Consider a network of $N$ agents indexed by $\{i\}_1^N$, where each agent holds an initial value $x_i(0) \in \mathbb{R}$. The goal of the average consensus problem is to drive all agents to asymptotically agree on the average of their initial values, represented by
\begin{align}
    \overline{x} = \frac{1}{N} \sum_{i=1}^N x_i(0),
\end{align}
using only local communication between neighbors \cite{olfati2007consensus}. The discrete-time average consensus protocol is given by
\begin{align}
    x_i(k+1) = x_i(k) + \epsilon \sum_{j \in \mathcal{N}_i} \left( x_j(k) - x_i(k) \right),
\end{align}
where $\epsilon$ is the consensus step size. In compact form, the update can be written as
\begin{align}
    \mathbf{x}(k+1) = \left(\mathbf{I}_N - \epsilon \mathbf{L}\right) \, \mathbf{x}(k),
\end{align}
where $\mathbf{L}$ denotes the graph Laplacian and $\mathbf{x}(k) \in \mathbb{R}^N$. For an undirected, connected graph with step size $0 < \epsilon < 1/\Delta$, where $\Delta$ represents the maximum degree of $\mathcal{G}$, the state of every agent converges to the average given by
\begin{align*}
    \lim_{k \to \infty} x_i(k) = \overline{x}, \ \forall i.
\end{align*}

In general, a larger step size $\epsilon$ (within the bound $1/\Delta$) tends to accelerate convergence, while a step size exceeding $1/\Delta$ leads to divergence. The average consensus algorithm can be extended to vector- or matrix-valued data by applying the algorithm independently to each entry. When $x_i(0) \in \mathbb{R}^m$, the compact form of the average consensus protocol is represented as
\begin{align}
    \mathbf{x}(k+1) = \left((\mathbf{I}_N - \epsilon \mathbf{L}) \otimes \mathbf{I}_m \right) \, \mathbf{x}(k),
\end{align}
where $\mathbf{x}(k) = \left[x_1^\top(k) \ \cdots \ x_N^\top(k) \right]^\top \in \mathbb{R}^{Nm}$ is the stacked state vector.


\section{Problem Formulation} \label{sec:problem}

This section formulates the target dynamics and sensor measurement model, states the assumptions and definitions required for the subsequent analysis, and specifies the objectives of the proposed distributed estimation framework.


\subsection{System Description}

Consider an area containing $N_T$ targets, monitored by a sensor network $\mathcal{G}(k)$ comprising $N$ agents. The dynamics of target $t$ are modeled by a discrete-time linear time-invariant (LTI) system as
\begin{align}\label{sys1}
    \mathbf{x}^t(k+1) &= \mathbf{F}\mathbf{x}^t(k) + \mathbf{G}\mathbf{w}^t(k),
\end{align}
where $\mathbf{x}^t(k) \in \mathbb{R}^{n}$ is the state vector of target $t$, $\mathbf{F} \in \mathbb{R}^{n\times n}$ is the state transition matrix, $\mathbf{w}^t(k) \in \mathbb{R}^{n}$ is the white process noise distributed as $\mathcal{N}(\mathbf{0},\mathbf{Q})$, and $\mathbf{G} \in \mathbb{R}^{n\times n}$ is the process noise input matrix. The sensor measurement model is described as follows
\begin{align}\label{meas1}
    \mathbf{z}_i^t(k) &= \mathbf{H}^t_i(k) \, \mathbf{x}^t(k) + \mathbf{v}_i^t(k),
\end{align}
where $\mathbf{z}_i^t(k) \in \mathbb{R}^{m}$ is the local measurement vector of target $t$ obtained by agent $i$, $\mathbf{H}^t_i(k) \in \mathbb{R}^{m\times n}$ is the local measurement matrix, and $\mathbf{v}_i^t(k) \in \mathbb{R}^{m}$ is the white measurement noise distributed as $\mathcal{N}(\mathbf{0},\mathbf{R}_i^t)$. The total number of measurements obtained by agent $i$ at time $k$ is denoted by $L_i(k)$, and these measurements are represented by $\{ \mathbf{z}_{i,r}(k) \}_{r=1}^{L_i(k)}$. Once the $r$-th measurement $\mathbf{z}_{i,r}(k)$ has been associated with target $t$, it is denoted by $\mathbf{z}_i^t(k)$. The number of associated measurements is represented by $N_i^a(k)$, where $N_i^a(k) \leq L_i(k)$.

Given the set of sensor measurements for target $t$ up to time $k$, denoted by $\mathcal{Z}_i^t(k) = \{\mathbf{z}_i^t(0), \mathbf{z}_i^t(1), \ldots, \mathbf{z}_i^t(k)\}$, the state estimates of target $t$ at agent $i$ are defined as
\begin{align*}
    \bar{\mathbf{x}}_i^t(k) &= \mathbb{E}\big[\mathbf{x}^t(k) | \mathcal{Z}_i^t(k-1)\big],\\
    \hat{\mathbf{x}}_i^t(k) &= \mathbb{E}\big[\mathbf{x}^t(k) | \mathcal{Z}_i^t(k)\big],
\end{align*}
where $\bar{\mathbf{x}}_i^t(k)$ and $\hat{\mathbf{x}}_i^t(k)$ are called the \textit{prior estimate} (or prediction) and the \textit{posterior estimate} (or updated estimate) of the state $\mathbf{x}^t(k)$ of target $t$, computed at agent $i$, respectively. The prior and posterior estimation errors are respectively defined as $\bar{\boldsymbol{\eta}}_i^t(k) = \bar{\mathbf{x}}_i^t(k) -\mathbf{x}^t(k)$ and $\hat{\boldsymbol{\eta}}_i^t(k) = \hat{\mathbf{x}}_i^t(k) -\mathbf{x}^t(k)$. The prior and posterior error covariance matrices are defined as
\begin{align*}
    \mathbf{P}_i^t(k) &= \mathbb{E}\left[ \bar{\boldsymbol{\eta}}_i^t(k) \,\big(\bar{\boldsymbol{\eta}}_i^t(k)\big)^\top \right],\\
    \mathbf{M}_i^t(k) &= \mathbb{E}\left[ \hat{\boldsymbol{\eta}}_i^t(k) \, \big(\hat{\boldsymbol{\eta}}_i^t(k)\big)^\top \right].
\end{align*}

An \textit{information matrix} is represented as the inverse of a covariance matrix, while an \textit{information vector} is described as the product of the information matrix and the state estimate vector.


\subsection{Assumptions and Definitions}

\begin{assumption}\label{ass:independent}
    The sequences $\{\mathbf{v}_i^t(k)\}$ are mutually independent for all agents $i$, and are independent of $\{\mathbf{w}^t(k)\}$ and the initial state $\mathbf{x}^t(0)$.
\end{assumption}

\begin{assumption}\label{ass:knowledge}
     $\mathbf{F}$, $\mathbf{G}$, $\mathbf{Q}$, $\mathbf{H}_i^t(k)$, $\mathbf{R}_i^t$, $N_T$, $\mathbf{P}_i^t(0)$ and $\bar{\mathbf{x}}_i^t(0)$ $\forall \{t\}_1^{N_T}$ for each agent $i$ are known.
\end{assumption}

\begin{assumption}\label{ass:constants}
    There exist constants $0 < \underline{p} \leq \overline{p} < \infty$ such that
    \begin{align*}
        \underline{p}\,\mathbf{I}_n \ \preceq \ \mathbf{P}_i^t(k) \ \preceq \ \overline{p}\,\mathbf{I}_n, \quad \forall i,\ \forall t,\ \forall k \geq 0.
    \end{align*}
\end{assumption}

\begin{assumption}\label{ass:observable}
    The pair $\left(\mathbf{F}, \mathbf{H}^t(k)\right)$ is observable, where $\mathbf{H}^t(k) \triangleq \left[\mathbf{H}^{t\top}_1(k) \ \mathbf{H}^{t\top}_2(k) \ \cdots \ \mathbf{H}^{t\top}_N(k) \right]^\top$. No assumption is imposed on the observability of the individual pairs $\left(\mathbf{F}, \mathbf{H}^t_i(k)\right)$.
\end{assumption}

\begin{definition}[Oblivious agent] \label{def:oblivious}
    Let $\mathcal{D}^t(k) \subseteq \mathcal{V}$ denote the set of agents directly observing target $t$ at time $k$. Then, an agent $i$ is called \textit{oblivious} to target $t$, if it belongs to the set
    \begin{align*}
    \mathcal{O}^t(k) = \{i \in \mathcal{V} \ | \ \mathcal{J}_i(k) \cap \mathcal{D}^t(k) = \varnothing  \}.
    \end{align*}
\end{definition}

Definition~\ref{def:oblivious} implies that for an agent to be \textit{oblivious}, neither the agent nor any of its neighbors senses the target. It is worth noting that alternative terms appear in the literature for such agents, e.g., \textit{naive node} in \cite{kamal2011generalized} and \textit{blind agent} in \cite{wang2017convergence}.

\begin{assumption}\label{ass:target}
    For any agent $i$ that does not directly observe target $t$, i.e., $ i \in \mathcal{V} \setminus \mathcal{D}^t(k)$, $\mathbf{H}^t_i(k) = \mathbf{0}$ and $(\mathbf{R}_i^t)^{-1} = \mathbf{0}$.
\end{assumption}

\begin{definition}[Consistency \cite{julier1997non}] \label{def:consistency}
    Let $\hat{\mathbf{x}}$ be an estimate of a random variable $\mathbf{x}$ with an associated reported error covariance $\mathbf{P}$. If $\acute{\mathbf{P}} \triangleq \mathbb{E}\big[(\mathbf{x} - \hat{\mathbf{x}})(\mathbf{x} - \hat{\mathbf{x}})^\top\big]$ is the true error covariance of $\hat{\mathbf{x}}$, then the pair $(\hat{\mathbf{x}}, \mathbf{P})$ is called \textit{consistent} if $\mathbf{P} \succeq \acute{\mathbf{P}}$. 
\end{definition}

Definition~\ref{def:consistency} states that a consistent estimate must have a reported covariance that is no less than the true error covariance, i.e., the uncertainty must not be underestimated.

\begin{assumption}\label{ass:jc}
    There exists $B \in \mathbb{Z}^+$ such that, for every $k \geq 0$, the graph sequence $\{\mathcal{G}(t)\}_{t=k}^{k+B-1}$ is jointly connected over the interval $[k, k+B-1]$.
\end{assumption}


\subsection{Problem Statement}

Consider $N_T$ targets with dynamics given by Eq.~\eqref{sys1}, observed by $N$ agents through Eq.~\eqref{meas1} and communicating over the undirected switching graph $\mathcal{G}(k)=(\mathcal{V},\mathcal{E}(k))$. As illustrated in Fig.~\ref{fig:oblivious_agents}, the set of agents observing a given target may vary with time, so an agent may be oblivious to target $t$ at some time instants.

\begin{problem}
Design a distributed Kalman filter that enables every agent to recursively compute a local estimate of each target state using only (i) its local measurements and (ii) information received from its neighbors. The designed filter must satisfy the following requirements: (I) no agent requires knowledge of any global network parameter, (II) the filter remains applicable under switching communication topologies, and (III) each agent communicates with its neighbors at most once per time step.
\end{problem}

\begin{figure}
    \centering
    \includegraphics[width=0.55\textwidth]{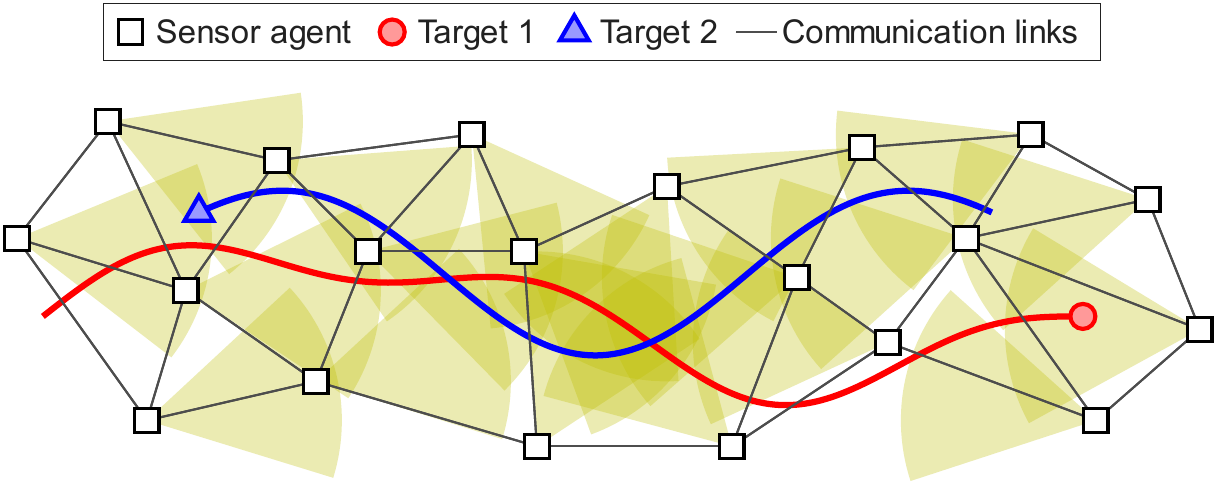}
    \caption{An illustration of a network of agents, with a limited FOV and constrained communication bandwidth, tracking two targets.}
    \label{fig:oblivious_agents}
\end{figure}

To reduce notational clutter while preserving agent and target indices, the time index $k$ is omitted from all equations throughout the remainder of this paper, unless otherwise stated. Furthermore, the updated value of any variable $\mathbf{x}$ is denoted by $\mathbf{x}^+$. For example, $\mathbf{x}(k+1) = f\left( \mathbf{x}(k), \mathbf{w}(k) \right)$ is expressed as $\mathbf{x}^+ = f(\mathbf{x}, \mathbf{w})$.


\section{Background} \label{sec:background}

This section reviews the centralized Kalman filter as a performance benchmark and presents the corresponding distributed estimation framework upon which the proposed algorithm is built.


\subsection{Centralized Estimation}

The posterior state estimate for target $t$ at sensor agent $i$ is computed using the Kalman filter \cite{kalman1960new} as
\begin{align}
    \hat{\mathbf{x}}_i^t &= \bar{\mathbf{x}}_i^t + \mathbf{K}_i^t\left( \mathbf{z}_i^t - \mathbf{H}_i^t\bar{\mathbf{x}}_i^t \right),\\
    \mathbf{K}_i^t &= \mathbf{P}_i^t \left(\mathbf{H}_i^t \right)^\top \left( \mathbf{H}_i^t \mathbf{P}_i^t \left(\mathbf{H}_i^t\right)^\top + \mathbf{R}_i^t \right)^{-1},
\end{align}
where $\mathbf{K}_i^t$ is the Kalman gain. To simplify notation, the target index $t$ is omitted from the remaining equations in this section, unless otherwise stated.

In a centralized Kalman filter (CKF), all local sensor measurements are collected at a fusion center, stacked into a single measurement vector, and jointly processed using the standard Kalman filter recursion to produce an optimal global state estimate under linear Gaussian assumptions. The centralized measurement vector, measurement matrix, and measurement noise covariance matrix are defined as follows, respectively,
\begin{align}
    \mathbf{z}_c &= \begin{bmatrix}
        \mathbf{z}_1^\top & \mathbf{z}_2^\top & \cdots & \mathbf{z}_N^\top
    \end{bmatrix}^\top \in \mathbb{R}^{Nm},\\
    \mathbf{H}_c &= \begin{bmatrix}
        \mathbf{H}_1^\top & \mathbf{H}_2^\top & \cdots & \mathbf{H}_N^\top
    \end{bmatrix}^\top \in \mathbb{R}^{Nm \times n},\\
    \mathbf{R}_c &= \text{diag}\big( \mathbf{R}_1, \mathbf{R}_2, \ldots, \mathbf{R}_N \big) \in \mathbb{R}^{Nm \times Nm},
\end{align}
The centralized posterior estimate is given by $\hat{\mathbf{x}}_c = \bar{\mathbf{x}}_c + \mathbf{K}_c \big(\mathbf{z}_c - \mathbf{H}_c\bar{\mathbf{x}}_c \big)$ with $\mathbf{K}_c = \mathbf{P}_c \mathbf{H}_c^\top \left(\mathbf{H}_c\mathbf{P}_c\mathbf{H}_c^\top + \mathbf{R}_c \right)^{-1}$, where $\bar{\mathbf{x}}_c$ and $\mathbf{P}_c$ are the centralized prior estimate and the centralized prior error covariance matrix, respectively. Applying the \textit{matrix inversion lemma} \cite{woodbury1950inverting}, the centralized posterior covariance and state estimate can be expressed in the \textit{information form} \cite{mutambara2019decentralized} as
\begin{align}
    \mathbf{M}_c &= \left( \mathbf{P}_c^{-1} + \mathbf{H}_c^\top \mathbf{R}_c^{-1} \mathbf{H}_c \right)^{-1} = \left( \mathbf{P}_c^{-1} + \sum_{i=1}^N \mathbf{H}_i^\top \mathbf{R}_i^{-1} \mathbf{H}_i \right)^{-1},\\
    \hat{\mathbf{x}}_c &= \mathbf{M}_c \left( \mathbf{P}_c^{-1}\bar{\mathbf{x}}_c + \mathbf{H}_c^\top \mathbf{R}_c^{-1} \mathbf{z}_c \right) = \mathbf{M}_c \left( \mathbf{P}_c^{-1}\bar{\mathbf{x}}_c + \sum_{i=1}^N \mathbf{H}_i^\top \mathbf{R}_i^{-1} \mathbf{z}_i \right).
\end{align}

\begin{remark}
    Computation of $\hat{\mathbf{x}}_c$ requires information to be transmitted from each agent in the network to a common fusion center. Due to communication and computational constraints, this method may become impractical as the network size increases.
\end{remark}


\subsection{Distributed Estimation}

In contrast to the centralized approach, distributed filtering algorithms assume that each agent has access to the prior estimates, prior covariances, measurements, measurement matrices, and measurement noise covariances of only its neighbors at each time step.

Neglecting neighbor data, the local Kalman filter computes the posterior covariance and state estimate for agent $i$, respectively \cite{olfati2007distributed}, as
\begin{align}
    \mathbf{M}_i &= \left(\mathbf{P}_i^{-1} + \mathbf{H}_i^\top \mathbf{R}_i^{-1} \mathbf{H}_i\right)^{-1}, \label{eq1} \\
    \hat{\mathbf{x}}_i &= \mathbf{M}_i \left(\mathbf{P}_i^{-1} \bar{\mathbf{x}}_i + \mathbf{H}_i^\top \mathbf{R}_i^{-1} \mathbf{z}_i\right). \label{eq2} 
\end{align}

Since $\bar{\mathbf{x}}_i = \mathbf{x} + \bar{\boldsymbol{\eta}}_i$, a merged state and sensor model can be written as $\breve{\mathbf{z}}_i = \breve{\mathbf{H}}_i \mathbf{x} + \breve{\mathbf{v}}_i$, where 
\begin{align*}
\breve{\mathbf{z}}_i \triangleq \big[\bar{\mathbf{x}}_i^\top \ \mathbf{z}_i^\top \big]^\top, \ \ \ \breve{\mathbf{H}}_i \triangleq \big[ \mathbf{I}_n^\top \ \mathbf{H}_i^\top \big]^\top, \ \ \ \breve{\mathbf{v}}_i \triangleq \big[\bar{\boldsymbol{\eta}}_i^\top \ \mathbf{v}_i^\top \big]^\top. 
\end{align*}
The noise covariance matrix of this model is defined as $\breve{\mathbf{R}}_i \triangleq \text{diag}\big(\mathbf{P}_i, \mathbf{R}_i \big)$. Substituting the above definitions into Eqs.~\eqref{eq1} and \eqref{eq2}, the local posterior covariance and state estimate can be written compactly as
\begin{align}
    \mathbf{M}_i = \left( \breve{\mathbf{H}}_i^\top  \breve{\mathbf{R}}_i^{-1} \breve{\mathbf{H}}_i \right)^{-1}, \quad \hat{\mathbf{x}}_i = \mathbf{M}_i \left( \breve{\mathbf{H}}_i^\top  \breve{\mathbf{R}}_i^{-1} \breve{\mathbf{z}}_i \right). \label{eq3}
\end{align}

Incorporating the available neighbor data, these merged vectors and matrices can be redefined as
\begin{align*}
    \breve{\mathbf{z}}_i &\triangleq \left[
        \bar{\mathbf{x}}_i^\top \ \bar{\mathbf{x}}_{j_1}^\top \ \cdots \ \bar{\mathbf{x}}_{j_{|\mathcal{N}_i|}}^\top \ \mathbf{z}_i^\top \ \mathbf{z}_{j_1}^\top \ \cdots \ \mathbf{z}_{j_{|\mathcal{N}_i|}}^\top \right]^\top,\\
    \breve{\mathbf{H}}_i &\triangleq \left[
        \mathbf{I}_n^\top \ \mathbf{I}_n^\top \ \cdots \ \mathbf{I}_n^\top \ \mathbf{H}_i^\top \ \mathbf{H}_{j_1}^\top \ \cdots \ \mathbf{H}_{j_{|\mathcal{N}_i|}}^\top
    \right]^\top,\\
    \breve{\mathbf{v}}_i &\triangleq \left[
        \bar{\boldsymbol{\eta}}_i^\top \ \bar{\boldsymbol{\eta}}_{j_1}^\top \ \cdots \ \bar{\boldsymbol{\eta}}_{j_{|\mathcal{N}_i|}}^\top \ \mathbf{v}_i^\top \ \mathbf{v}_{j_1}^\top \ \cdots \ \mathbf{v}_{j_{|\mathcal{N}_i|}}^\top
    \right]^\top.
\end{align*}

The merged noise covariance matrix is now defined as $\breve{\mathbf{R}}_i \triangleq \text{diag}\left(\boldsymbol{\mathsf{P}}_i, \mathbf{R}_i, \mathbf{R}_{j_1}, \ldots, \mathbf{R}_{j_{|\mathcal{N}_i|}}\right)$. Using $\mathbf{P}_{ij} = \mathbb{E}\left[\bar{\boldsymbol{\eta}}_i \bar{\boldsymbol{\eta}}_j^\top\right]$, $\boldsymbol{\mathsf{P}}_i \in \mathbb{R}^{|\mathcal{J}_i|n \times |\mathcal{J}_i|n}$ is given by
\begin{align*}
    \boldsymbol{\mathsf{P}}_i \triangleq \begin{bmatrix}
        \mathbf{P}_i & \mathbf{P}_{ij_1} & \cdots & \mathbf{P}_{ij_{|\mathcal{N}_i|}}\\
        \mathbf{P}_{j_1i} & \mathbf{P}_{j_1} & \cdots & \mathbf{P}_{j_1j_{|\mathcal{N}_i|}}\\
        \vdots & \vdots & \ddots & \vdots\\
        \mathbf{P}_{j_{|\mathcal{N}_i|}i} & \mathbf{P}_{j_{|\mathcal{N}_i|}j_1} & \cdots & \mathbf{P}_{j_{|\mathcal{N}_i|}}
    \end{bmatrix}.
\end{align*}

The local posterior covariance and state estimate in Eq.~\eqref{eq3} can now be represented as
\begin{align}
    \mathbf{M}_i &= \left(\boldsymbol{\mathsf{H}}_i^\top \boldsymbol{\mathsf{P}}_i^{-1} \boldsymbol{\mathsf{H}}_i + \sum_{j \in \mathcal{J}_i} \mathbf{H}_j^\top \mathbf{R}_j^{-1} \mathbf{H}_j\right)^{-1}, \label{eq4} \\
    \hat{\mathbf{x}}_i &= \mathbf{M}_i \left(\boldsymbol{\mathsf{H}}_i^\top \boldsymbol{\mathsf{P}}_i^{-1} \bar{\boldsymbol{\mathsf{x}}}_i + \sum_{j \in \mathcal{J}_i} \mathbf{H}_j^\top \mathbf{R}_j^{-1} \mathbf{z}_j\right), \label{eq5} 
\end{align}
where $\boldsymbol{\mathsf{H}}_i \triangleq \left[ \mathbf{I}_n^\top \ \cdots \ \mathbf{I}_n^\top \right]^\top \in \mathbb{R}^{|\mathcal{J}_i|n \times n}$ and $\bar{\boldsymbol{\mathsf{x}}}_i \triangleq \left[\bar{\mathbf{x}}_i^\top \ \bar{\mathbf{x}}_{j_1}^\top \ \cdots \ \bar{\mathbf{x}}_{j_{|\mathcal{N}_i|}}^\top \right]^\top$. The expressions in Eqs.~\eqref{eq4} and \eqref{eq5} constitute a theoretically optimal rule for fusing multiple correlated estimates in a distributed setting. With sufficient time and appropriate connectivity conditions, these local posterior covariances $\mathbf{M}_i$ converge to the centralized posterior covariance $\mathbf{M}_c$.

Despite this optimality, the formulation is not practically realizable in a distributed setting. Specifically, each agent would need to maintain not only its local prior covariance $\mathbf{P}_i$ and those of its neighbors $\mathbf{P}_j$ for all $j \in \mathcal{N}_i$, but also all pairwise cross-covariances $\mathbf{P}_{ij}$, whose computation depends on information beyond the agent's local neighborhood and therefore cannot be computed locally. As a result, this scheme requires substantial communication and computational overhead (inversion of $\boldsymbol{\mathsf{P}}_i$), which may lead to it becoming impractical as the network size grows. Therefore, practical distributed filtering algorithms trade optimality for scalability.

Some distributed filters approximate $\boldsymbol{\mathsf{H}}_i^\top \boldsymbol{\mathsf{P}}_i^{-1} \boldsymbol{\mathsf{H}}_i$ with $\sum_{j \in \mathcal{J}_i} \mathbf{P}_j^{-1}$ \cite{wang2014information}. The local posterior covariance and state estimate in Eqs.~\eqref{eq4} and \eqref{eq5}, respectively, can then be expressed as
\begin{align}
    \mathbf{M}_i &= \left( \sum_{j \in \mathcal{J}_i} \mathbf{P}_j^{-1} + \sum_{j \in \mathcal{J}_i} \mathbf{H}_j^\top \mathbf{R}_j^{-1} \mathbf{H}_j \right)^{-1}, \label{eq6} \\
    \hat{\mathbf{x}}_i &= \mathbf{M}_i \left( \sum_{j \in \mathcal{J}_i} \mathbf{P}_j^{-1} \bar{\mathbf{x}}_j + \sum_{j \in \mathcal{J}_i} \mathbf{H}_j^\top \mathbf{R}_j^{-1} \mathbf{z}_j\right). \label{eq7} 
\end{align}

This formulation is computationally tractable, and all terms can be computed in a distributed setup, though it is suboptimal. However, the summation of prior information matrices in Eq.~\eqref{eq6} introduces a further complication. Agents with larger neighborhoods aggregate a disproportionate amount of prior information, causing $\mathbf{M}_i$ to vary with agent degree and to diverge from $\mathbf{M}_c$. To address this, the sum is replaced with an average consensus protocol \cite{kamal2011generalized, kamal2013information} as
\begin{align}
    \widetilde{\mathbf{P}}_i^{-1} \triangleq \mathbf{P}_i^{-1} 
	+ \epsilon \sum_{j \in \mathcal{N}_i} \left( \mathbf{P}_j^{-1} - \mathbf{P}_i^{-1} \right). \label{tildeP}
\end{align}
Although a single iteration does not guarantee convergence to the global average $\frac{1}{N}\sum_{i=1}^{N} \mathbf{P}_i^{-1}$ within a given time step, its repeated application across successive time steps progressively reduces the disagreement among the local posterior covariances $\mathbf{M}_i$. Under appropriate connectivity conditions, the local covariances are then driven toward the centralized value $\mathbf{M}_c$.

Note that consensus on the information vectors $\mathbf{P}_j^{-1}\bar{\mathbf{x}}_j$ from Eq.~\eqref{eq7} is equivalent to a weighted-consensus on the prior state estimates $\bar{\mathbf{x}}_j$, where each agent's contribution is weighted by its own prior information matrix $\mathbf{P}_j^{-1}$, defined as
\begin{align}
    \widetilde{\mathbf{P}}_i^{-1} \widetilde{\mathbf{x}}_i \triangleq \mathbf{P}_i^{-1}\bar{\mathbf{x}}_i 
	+ \epsilon \sum_{j \in \mathcal{N}_i} \left( \mathbf{P}_j^{-1}\bar{\mathbf{x}}_j - \mathbf{P}_i^{-1}\bar{\mathbf{x}}_i \right).
\end{align}

This formulation inherently accommodates oblivious agents. These agents carry no measurement information about the target, resulting in a large prior covariance $\mathbf{P}_i$ and a near-zero information matrix $\mathbf{P}_i^{-1}$. Consequently, such an agent's contribution to the weighted average is automatically suppressed, preventing uninformed prior estimates from degrading the fused estimate.

\begin{remark}
    For the consensus step to converge, $\epsilon$ must satisfy $\epsilon \in (0,1/\Delta)$. However, in a dynamic network where the topology changes with time, $\Delta$ is neither fixed nor locally observable. The addition or removal of edges can change $\Delta$, causing a previously admissible choice of $\epsilon$ to violate the convergence condition. If $\epsilon$ is set conservatively based on the largest possible $\Delta$, the resulting $\epsilon$ may be unnecessarily small, degrading the convergence rate. Thus, a fixed $\epsilon$ may lead to an impractical choice.
\end{remark}


\section{Information-Weighted Distributed Kalman Filter for Multi-Target Tracking} \label{sec:proposed}

This section derives the proposed filter, introduces the dynamic, locally computable weighting strategy, and presents the data association procedure used for multi-target tracking.


\subsection{Distributed Consensus Update}

To address the limitations associated with a fixed consensus step size, we introduce dynamic consensus weights that do not require global topology information and can be computed locally. Since each agent has access to the prior covariances of its neighbors, we first associate a scalar uncertainty measure with each neighboring agent's prior estimate. As $\mathbf{P}_{j}^t \in \mathbb{S}^n_{++}$ with eigenvalues $\lambda_1 \big(\mathbf{P}_j^t\big) \geq \cdots \geq \lambda_n \big(\mathbf{P}_j^t\big) > 0$, the expected squared Euclidean norm of the prior estimation error satisfies
\begin{align}
    \mathbb{E}\left[ \| \mathbf{x}^t - \bar{\mathbf{x}}_j^t \|^2 \right] = \sum_{r=1}^{n}\lambda_r \big(\mathbf{P}_j^t\big) = \trace \big(\mathbf{P}_j^t \big),
\end{align}
which constitutes the total mean squared error (MSE) of agent $j$'s prior and is equal to the trace of the prior covariance. Its reciprocal, defined as
\begin{align*}
    \rho_j^t \triangleq \left(\trace \big(\mathbf{P}_j^t \big) \right)^{-1},
\end{align*}
therefore provides a scalar measure of the prior reliability of agent $j$, where larger values indicate more informative estimates.

To preserve reciprocity in the information exchange over an undirected communication graph, we define a symmetric information metric of an edge $(i,j)$ as the sum of the information values of its two endpoints, represented by
\begin{align}\label{Eq:phi_ij_1}
	\phi_{ij}^t \triangleq \rho_i^t + \rho_j^t = \left(\trace \big(\mathbf{P}_i^t \big) \right)^{-1} + \left(\trace \big(\mathbf{P}_j^t \big) \right)^{-1}.
\end{align}
This quantity measures the joint reliability of the prior estimates of agents $i$ and $j$. Compared with other scalar metrics such as the determinant or spectral norm, the trace offers a lower computational burden, requiring only the addition of diagonal entries, and is numerically more stable.

The normalized weight assigned to each agent $j$ is defined as
\begin{align} \label{weights}
	w_{ij}^t \triangleq \begin{cases}
	    \dfrac{\phi_{ij}^t}{\sum_{r \in \mathcal{J}_i} \phi_{ir}^t} &,\ \text{if } j \in \mathcal{J}_i,\\[1em]
        \quad \quad 0 &,\ \text{if } j \notin \mathcal{J}_i.
	\end{cases}
\end{align}
Using these weights, agent $i$ computes its weighted-consensus estimate as
\begin{align}
	\left(\widetilde{\mathbf{P}}_i^t \right)^{-1} &\triangleq \big(\mathbf{P}_i^t\big)^{-1} + \sum_{j \in \mathcal{N}_i} w_{ij}^t \left( \big(\mathbf{P}_j^t\big)^{-1} - \big(\mathbf{P}_i^t\big)^{-1} \right), \label{tildePwij} \\
	\left(\widetilde{\mathbf{P}}_i^t \right)^{-1} \widetilde{\mathbf{x}}_i^t &\triangleq \big(\mathbf{P}_i^t\big)^{-1} \bar{\mathbf{x}}_i^t + \sum_{j \in \mathcal{N}_i} w_{ij}^t \left( \big(\mathbf{P}_j^t\big)^{-1} \bar{\mathbf{x}}_j^t - \big(\mathbf{P}_i^t\big)^{-1} \bar{\mathbf{x}}_i^t \right). \label{tildePwx}
\end{align}

The measurement update equations of the proposed filter design, henceforth referred to as the Information-Weighted Distributed Kalman Filter (IWDKF), are given by
\begin{align}
    \mathbf{M}_i^t &= \left( \left(\widetilde{\mathbf{P}}_i^t \right)^{-1} + \sum_{j \in \mathcal{J}_i} \big(\mathbf{H}_j^t\big)^\top \big(\mathbf{R}_j^t\big)^{-1} \mathbf{H}^t_j \right)^{-1},\\
	\hat{\mathbf{x}}_i^t &= \mathbf{M}_i^t \left( \left(\widetilde{\mathbf{P}}^{t}_i \right)^{-1}\widetilde{\mathbf{x}}_i^t + \sum_{j \in \mathcal{J}_i} \big(\mathbf{H}_j^t\big)^\top \big(\mathbf{R}_j^t\big)^{-1} \mathbf{z}^t_j \right). \label{x_update}
\end{align}
Therefore, the IWDKF update enables fully distributed state estimation using only local communication and locally computable consensus weights.

\begin{remark}
    Equations~\eqref{tildePwij}-\eqref{x_update} assume that each measurement $\mathbf{z}_j^t$ and each received prior estimate $\bar{\mathbf{x}}_j^t$ are already associated with target $t$. As their associations are not known a priori, both must be established before the update is applied.
\end{remark}


\subsection{Data Association}

In a multi-target environment, sensors produce unlabeled measurements, making it unclear which measurement originates from which target. Additionally, the prior estimates received from neighbors are indexed by each agent's local target labels, which need not be the same across agents. To address these problems, the Hungarian algorithm \cite{munkres1957algorithms} (also known as the Kuhn-Munkres algorithm) is widely adopted. It provides the optimal assignment for a linear cost formulation, and its polynomial-time complexity makes it suitable for real-time implementation. In the proposed framework, the Hungarian algorithm serves two purposes: to associate predicted states with incoming measurements, and to associate predicted states across different agents.

In tracking applications, the association cost matrix is commonly defined as a function of the innovation for each state–measurement pair. For example, the association cost between the predicted state $\bar{\mathbf{x}}_i^t$ of target $t$ and the $r$-th measurement $\mathbf{z}_{i,r}$ at agent $i$ is given by
\begin{equation}
    c_{tr} = \big( \mathbf{z}_{i,r} - \mathbf{H}^t_i\bar{\mathbf{x}}_i^t \big)^\top \big(\mathbf{S}_i^t\big)^{-1} \big( \mathbf{z}_{i,r} - \mathbf{H}^t_i\bar{\mathbf{x}}_i^t \big),
\end{equation}
where $\mathbf{S}_i^t = \big( \mathbf{H}_i^t \, \mathbf{P}_i^t \, (\mathbf{H}_i^t)^\top + \mathbf{R}_i^t \big)$ is the innovation covariance. Once the optimal associations have been determined, the measurement assigned to target $t$ at agent $i$ is henceforth denoted by $\mathbf{z}_i^t$. If target $t$  is not associated with any measurement at agent $i$, then $\mathbf{H}_i^t = \mathbf{0}$ and $\big(\mathbf{R}_i^t\big)^{-1} = \mathbf{0}$. This step is referred to as measurement-to-track association.

Similarly, the association cost between the predicted state $\bar{\mathbf{x}}_i^t$ of target $t$ at agent $i$ and $\bar{\mathbf{x}}_p^q$ of target $q$ at agent $p$ is represented by
\begin{align}
    c_{tq} = \big(\bar{\mathbf{x}}_i^t - \bar{\mathbf{x}}_p^q \big)^\top \big(\mathbf{P}_i^t + \mathbf{P}_p^q \big)^{-1} \big(\bar{\mathbf{x}}_i^t - \bar{\mathbf{x}}_p^q \big).
\end{align}
This step is referred to as track-to-track association. The complete IWDKF design, encompassing both the prediction and update steps, is summarized in Algorithm~\ref{alg1}.


\begin{algorithm}\small
\caption{Information-Weighted Distributed Kalman Filter}\label{alg1}
\begin{algorithmic}[1]

\Require $\bar{\mathbf{x}}_i^{t}$, $\mathbf{P}_i^{t}$, $\mathbf{H}_i^t$, $\mathbf{R}_i^t$, $\forall \{t\}_1^{N_T}$; $\mathbf{F}$, $\mathbf{G}$, $\mathbf{Q}$

\State Each agent $i$ executes the following:

\For{$k = 1,2,\dots,K-1$}
    \State Acquire local measurements $\{ \mathbf{z}_{i,r} \}_{r=1}^{L_i}$.
    
    \State Perform measurement-to-track association using the Hungarian algorithm.

    \State Compute local measurement information vector and matrix:
    \begin{align*}
        \mathbf{u}_i^t &= \left(\mathbf{H}^t_i \right)^\top \left(\mathbf{R}^t_i\right)^{-1} \mathbf{z}_i^t,\\
        \mathbf{U}_i^t &= \left(\mathbf{H}^t_i \right)^\top \left(\mathbf{R}_i^t \right)^{-1} \mathbf{H}^t_i, \ \forall \{t\}_{1}^{N_i^a}.
    \end{align*}

    \State Transmit $\mathcal{M}_i = \{\mathbf{u}_i^t, \mathbf{U}_i^t, \bar{\mathbf{x}}_i^t, \mathbf{P}_i^t \}$ to neighbors.

    \State Obtain messages $\mathcal{M}_j \ \forall j \in \mathcal{N}_i$.

    \State Perform track-to-track association using the Hungarian algorithm. Re-index received tracks to align with local target labels.

    \State Fuse measurement information vectors and matrices:
    \begin{align*}
        \mathbf{y}_i^t = \sum_{j \in \mathcal{J}_i} \mathbf{u}_j^t \ , \quad \mathbf{Y}_i^t = \sum_{j \in \mathcal{J}_i} \mathbf{U}_j^t \ , \ \forall \{t\}_{1}^{N_i^a}.
    \end{align*}

    \State Compute consensus weights:
    \begin{align*}
        \rho_i^t &= \left( \trace\big(\mathbf{P}_i^t\big) \right)^{-1},\ \ \phi_{ij}^t = \rho_i^t + \rho_j^t, \\
        w_{ij}^t &= \dfrac{\phi_{ij}^t}{\sum_{r \in \mathcal{J}_i} \phi_{ir}^t} \ \forall j \in \mathcal{J}_i.
    \end{align*}

    \State Compute weighted-consensus estimate:
    \begin{align*}
        \left(\widetilde{\mathbf{P}}_i^t \right)^{-1} &= \big(\mathbf{P}_i^t\big)^{-1} + \sum_{j \in \mathcal{N}_i} w_{ij}^t \left( \big(\mathbf{P}_j^t\big)^{-1} - \big(\mathbf{P}_i^t\big)^{-1} \right), \\
	    \left(\widetilde{\mathbf{P}}_i^t \right)^{-1} \widetilde{\mathbf{x}}_i^t &= \big(\mathbf{P}_i^t\big)^{-1} \bar{\mathbf{x}}_i^t + \sum_{j \in \mathcal{N}_i} w_{ij}^t \left( \big(\mathbf{P}_j^t\big)^{-1} \bar{\mathbf{x}}_j^t - \big(\mathbf{P}_i^t\big)^{-1} \bar{\mathbf{x}}_i^t \right).
    \end{align*}
    
    \State Compute the posterior state estimate:
    \begin{align*}
        \mathbf{M}_i^t &= \left( \left(\widetilde{\mathbf{P}}^{t}_i \right)^{-1} + \mathbf{Y}_i^t \right)^{-1},\\
        \hat{\mathbf{x}}_i^t &= \mathbf{M}_i^t \left( \left(\widetilde{\mathbf{P}}^{t}_i\right)^{-1}\widetilde{\mathbf{x}}_i^t + \mathbf{y}_i^t \right).
    \end{align*}

    \State Propagate the state estimate for the next time step:
    \begin{align*}
        \mathbf{P}_i^{t,+} &\leftarrow \mathbf{F} \mathbf{M}_i^t \mathbf{F}^\top + \mathbf{G} \mathbf{Q} \mathbf{G}^\top,\\ \bar{\mathbf{x}}_i^{t,+} &\leftarrow \mathbf{F} \hat{\mathbf{x}}_i^t.
    \end{align*}
    
\EndFor

\Ensure $\hat{\mathbf{x}}_i^t$, $\mathbf{M}_i^t$, $\bar{\mathbf{x}}_i^{t,+}$, $\mathbf{P}_i^{t,+}$, $\forall \{t\}_1^{N_T}$

\end{algorithmic}
\end{algorithm}


\section{Convergence and Consistency Analysis} \label{sec:convergence}

This section establishes the positive definiteness of the weighted-consensus state covariance matrix, analyzes the convergence properties of the proposed consensus protocol, and shows that the consensus step preserves the consistency property. To simplify notation, the target index $t$ is omitted from the equations in this section, unless otherwise stated.


\subsection{Positive Definiteness of the Weighted Consensus Covariance Matrix}

Considering the average consensus protocol, the expression for $\widetilde{\mathbf{P}}_i$, as described in Eq.~\eqref{tildeP}, can be reformulated as
\begin{align}
    \widetilde{\mathbf{P}}_i^{-1} &= \mathbf{P}_i^{-1} + \epsilon \sum_{j \in \mathcal{N}_i} \mathbf{P}_j^{-1} - \epsilon \sum_{j \in \mathcal{N}_i} \mathbf{P}_i^{-1} \nonumber \\
    &= \mathbf{P}_i^{-1} + \epsilon \sum_{j \in \mathcal{N}_i} \mathbf{P}_j^{-1} - \epsilon \, |\mathcal{N}_i| \, \mathbf{P}_i^{-1} \nonumber \\
    &= \left( 1 - \epsilon \, |\mathcal{N}_i| \right)\mathbf{P}_i^{-1} + \epsilon \sum_{j \in \mathcal{N}_i} \mathbf{P}_j^{-1},
\end{align}
where $|\mathcal{N}_i|$ denotes the cardinality of $\mathcal{N}_i$. Moreover, $\Delta$ can be expressed as $\text{max}_i |\mathcal{N}_i|$. Thus, the condition $0 < \epsilon < 1/\Delta$ guarantees that $\left( 1 - \epsilon \, |\mathcal{N}_i| \right) > 0$ for all $i$. Consequently, $\widetilde{\mathbf{P}}_i^{-1}$  is a positive linear combination of the symmetric positive definite matrices $\mathbf{P}_i^{-1}$ and $\mathbf{P}_j^{-1}$ for all $j \in \mathcal{N}_i$, which implies $\widetilde{\mathbf{P}}_i^{-1} \in \mathbb{S}^n_{++}$. This is a crucial requirement for the convergence of the filter. The following theorem establishes this positive definiteness property within the IWDKF framework.

\begin{theorem} \label{thm:positive-def}
    Consider a network of $N$ sensor agents, where each agent $i$ has a prior covariance matrix $\mathbf{P}_i \in \mathbb{S}^n_{++}$. If the weighted-consensus covariance is given by Eq.~\eqref{tildePwij} and the weights are defined by Eq.~\eqref{weights}, then the matrix $\widetilde{\mathbf{P}}_i \in \mathbb{S}^n_{++}$.
\end{theorem}

\begin{proof}
    Since each $\mathbf{P}_i \in \mathbb{S}^n_{++}$, it follows that $(\mathbf{P}_i)^{-1} \in \mathbb{S}^n_{++}$. By definition, $\phi_{ij}$ is the sum of the reciprocals of two strictly positive scalars    
    $\trace(\mathbf{P}_i)$ and $\trace(\mathbf{P}_j)$, given in Eq. \eqref{Eq:phi_ij_1}. Therefore, the normalizing denominator $\sum_{r \in \mathcal{J}_i} \phi_{ir}$ is strictly positive, which ensures that each weight $w_{ij}>0$. 
    
    By construction,
    \begin{align}\label{sumweights}
        \sum_{j \in \mathcal{J}_i} w_{ij} = w_{ii} + \sum_{j \in \mathcal{N}_i} w_{ij} = 1.
    \end{align}

    Using Eq.~\eqref{sumweights}, we can rewrite Eq.~\eqref{tildePwij} as
    \begin{align}
        \widetilde{\mathbf{P}}_i^{-1} &= \mathbf{P}_i^{-1} + \sum_{j \in \mathcal{N}_i} w_{ij} \mathbf{P}_j^{-1} - \mathbf{P}_i^{-1} \sum_{j \in \mathcal{N}_i} w_{ij} \nonumber \\
        &= \mathbf{P}_i^{-1} + \sum_{j \in \mathcal{N}_i} w_{ij} \mathbf{P}_j^{-1} - \mathbf{P}_i^{-1}(1 - w_{ii}) \nonumber \\
        &= w_{ii} \mathbf{P}_i^{-1} + \sum_{j \in \mathcal{N}_i} w_{ij} \mathbf{P}_j^{-1} \nonumber \\
        &= \sum_{j \in \mathcal{J}_i} w_{ij} \mathbf{P}_j^{-1}. \label{protocol}
    \end{align}

    Consider a vector $\mathbf{x} \in \mathbb{R}^n$ and $\mathbf{x} \neq \mathbf{0}$. Then, the following can be expressed as
    \begin{align*}
        \mathbf{x}^\top \widetilde{\mathbf{P}}_i^{-1} \mathbf{x} = \sum_{j\in\mathcal J_i} w_{ij} \, \mathbf{x}^\top \mathbf{P}_j^{-1} \mathbf{x}.
    \end{align*}
    For each $j\in\mathcal J_i$, we have $\mathbf{x}^\top \mathbf{P}_j^{-1} \mathbf{x} > 0$, because $\mathbf{P}_j^{-1} \in \mathbb{S}^n_{++}$ and each $w_{ij}>0$. Therefore, the weighted sum is strictly positive, i.e., $\mathbf{x}^\top \widetilde{\mathbf{P}}_i^{-1} \mathbf{x} > 0$, which implies $\widetilde{\mathbf{P}}_i^{-1} \in \mathbb{S}^n_{++}$. Consequently, $\widetilde{\mathbf{P}}_i \in \mathbb{S}^n_{++}$.
\end{proof}

\begin{remark}
    The result relies solely on the local assumption that the prior covariances in the inclusive neighborhood of agent $i$ are symmetric positive definite. No global network parameters, such as $\Delta$, are required to guarantee $\widetilde{\mathbf{P}}_i \in \mathbb{S}^n_{++}$. However, Theorem~\ref{thm:positive-def} is a property of each agent and does not imply agreement among agents, which must be established separately.
\end{remark}


\subsection{Convergence Analysis}

This subsection shows that the proposed weighting strategy preserves the asymptotic agreement property of the underlying consensus protocol. Accordingly, an inner consensus iteration index $\ell$ is introduced.

Let $\mathbf{X}_i(\ell)$ denote the information matrix maintained by agent $i$ after the $\ell$-th consensus iteration, with initialization $\mathbf{X}_i(0) = \mathbf{P}_i^{-1}$. The consensus protocol in Eq.~\eqref{protocol} can now be expressed as
\begin{align}
    \mathbf{X}_i(\ell+1) = \sum_{j \in \mathcal{J}_i(\ell)} w_{ij}(\ell)\,\mathbf{X}_j(\ell),
    \quad i=1,\ldots,N. \label{consensus_law}
\end{align}

After $M$ consensus iterations, the consensus information matrix is represented by
\begin{align}
    \widetilde{\mathbf{P}}_i^{-1} = \mathbf{X}_i(M),
\end{align}
which is subsequently used in the measurement update. 

Define the stacked vector $\mathbf{s(\ell)}$ as
\begin{align*}
    \mathbf{s}(\ell) \triangleq \left[ \operatorname{vec}(\mathbf{X}_1(\ell))^\top \ \cdots \ \operatorname{vec}(\mathbf{X}_N(\ell))^\top \right]^\top \in \mathbb{R}^{Nn^2}.
\end{align*}
The collective dynamics can then be written as
\begin{align} \label{collective}
    \mathbf{s}(\ell + 1) = \left(\mathbf{W}(\ell) \otimes \mathbf{I}_{n^2} \right) \, \mathbf{s}(\ell),
\end{align}
where $\mathbf{W}(\ell) = \left[w_{ij}(\ell)\right] \in \mathbb{R}^{N\times N}$, with $w_{ij}(\ell)>0$ if $(i,j) \in \mathcal{E}(\ell)$ or $i=j$, and $w_{ij}(\ell)=0$ otherwise. By construction, $\mathbf{W}(\ell)$ is row-stochastic with $w_{ii}(\ell)>0$ for every $i,\ell$. Iterating Eq.~\eqref{collective}, we obtain
\begin{align}
    \mathbf{s}(\ell) &= \big(\mathbf{W}(\ell-1) \otimes \mathbf{I}_{n^2}\big) \cdots \big(\mathbf{W}(0) \otimes \mathbf{I}_{n^2}\big) \, \mathbf{s}(0) \nonumber\\
    &= \big( \boldsymbol{\Phi}(\ell,0) \otimes \mathbf{I}_{n^2} \big) \, \mathbf{s}(0), \label{collective2}
\end{align}
where $\boldsymbol{\Phi}(\ell_2,\ell_1) \triangleq \mathbf{W}(\ell_2-1) \mathbf{W}(\ell_2-2) \cdots \mathbf{W}(\ell_1)$, when $\ell_2 > \ell_1$ and $\boldsymbol{\Phi}(\ell_2,\ell_1) \triangleq \mathbf{I}_N$ when $\ell_2 = \ell_1$. Note that $\boldsymbol{\Phi}(\ell_2,\ell_1)$ is row-stochastic for every $\ell_2 \geq \ell_1$, since it is a product of row-stochastic matrices.


\begin{lemma}
\label{lem:uniform-weight}
There exists $\alpha \triangleq \underline{p}/\big(N\overline{p}\big) > 0$ such that
\begin{align*}
    w_{ij}(\ell) \ \geq\ \alpha, \quad \forall j \in \mathcal{J}_i(\ell),\ \forall i,\ \forall \ell \geq 0.
\end{align*}
\end{lemma}

\begin{proof}
Since $\rho_i(\ell) = \left(\trace \big(\mathbf{P}_i(\ell) \big)\right)^{-1}$ and Assumption~\ref{ass:constants} gives $\trace \big(\mathbf{P}_i(\ell)\big) \in \big[n\underline{p}, n\overline{p}\big]$, we get
\begin{align*}
    \rho_i(\ell) \in \big[\rho_{\min}, \rho_{\max}\big] \triangleq \left[\frac{1}{n\overline{p}},\ \frac{1}{n\underline{p}}\right], \quad \forall \, i,\ell.
\end{align*}
Consequently, $\phi_{ij}(\ell) = \rho_i(\ell)+\rho_j(\ell) \in \big[2\rho_{\min}, 2\rho_{\max}\big]$ for all $i,j,\ell$. Since $|\mathcal{J}_i(\ell)| \leq N$, we can write
\begin{align*}
    w_{ij}(\ell) = \frac{\phi_{ij}(\ell)}{\sum_{r \in \mathcal{J}_i(\ell)}\phi_{ir}(\ell)} \ \geq\ \frac{2\rho_{\min}}{N 2\rho_{\max}} = \frac{\underline{p}}{N\overline{p}} \triangleq \alpha > 0,
\end{align*}
for every $j \in \mathcal{J}_i(\ell)$.
\end{proof}


\begin{lemma}[\cite{jadbabaie2003coordination}, Lemma 2]
\label{lem:jadbabaie-lem2}
Let $p \geq 2$, $\mathbf{D}_1,\ldots,\mathbf{D}_p$ be nonnegative $N{\times}N$ matrices with positive diagonal entries, and $\mu,\xi$ be the smallest and largest diagonal entries among them, respectively. Then, the following can be expressed as
\begin{align*}
    \mathbf{D}_1\cdots\mathbf{D}_p \ \geq\ \left(\frac{\mu^2}{2\xi}\right)^{p-1} \big(\mathbf{D}_1+\cdots+\mathbf{D}_p\big).
\end{align*}
\end{lemma}

\begin{proof}
    See \cite[Lemma~2]{jadbabaie2003coordination}.
\end{proof}


\begin{lemma}
\label{lem:block-bound}
For every $\ell \geq 0$,
\begin{align*}
    \boldsymbol{\Phi}(\ell{+}B,\ell) \ \geq\ \gamma\big(\mathbf{A}_{\mathrm{u}}^{(\ell)} + \mathbf{I}_N\big),
\end{align*}
where $\mathbf{A}_{\mathrm{u}}^{(\ell)}$ is the adjacency matrix of $\bigcup_{t=\ell}^{\ell+B-1}\mathcal{G}(t)$ and $\gamma \triangleq \alpha\big(\alpha^2/2\big)^{B-1} > 0$.
\end{lemma}

\begin{proof}
By Lemma~\ref{lem:uniform-weight}, each $\mathbf{W}(t)$, $t=\ell,\ldots,\ell{+}B{-}1$, has diagonal entries at least $\alpha$, and row-stochasticity gives an upper bound of $1$ on all entries. Lemma~\ref{lem:jadbabaie-lem2} then provides
\begin{align*}
    \boldsymbol{\Phi}(\ell{+}B,\ell) \ &\geq\ \left(\frac{\alpha^2}{2}\right)^{B-1} \sum_{t=\ell}^{\ell+B-1}\mathbf{W}(t),\\
    \ &\geq\ \gamma \sum_{t=\ell}^{\ell+B-1}\big(\mathbf{A}(t)+\mathbf{I}_N\big),
\end{align*}
using $\mathbf{W}(t) \geq \alpha \big(\mathbf{A}(t)+\mathbf{I}_N\big)$ from Lemma~\ref{lem:uniform-weight}. Now, since $\sum_t \mathbf{A}(t) \geq \mathbf{A}_{\mathrm{u}}^{(\ell)}$ and $\sum_t \mathbf{I}_N \geq \mathbf{I}_N$, we obtain
\begin{align*}
    \boldsymbol{\Phi}(\ell{+}B,\ell) \ \geq\ \gamma \sum_{t=\ell}^{\ell+B-1}\big(\mathbf{A}(t)+\mathbf{I}_N\big) \ \geq\ \gamma\big(\mathbf{A}_{\mathrm{u}}^{(\ell)} + \mathbf{I}_N\big).
\end{align*}
\end{proof}


\begin{lemma}
\label{lem:reachability}
Let $\mathcal{G}_1,\ldots,\mathcal{G}_{N-1}$ be connected simple graphs on $\{1,\ldots,N\}$ and $\mathbf{C}_h \triangleq \mathbf{I}_N + \mathbf{A}_{h}$, where $\mathbf{A}_h$ is the adjacency matrix associated with $\mathcal{G}_h$. There exists $\beta = \beta(N) > 0$ such that
\begin{align*}
    \mathbf{C}_{N-1}\cdots\mathbf{C}_1 \ \geq\ \beta\,\mathbf{1}\mathbf{1}^\top,
\end{align*}
for every such sequence.
\end{lemma}

\begin{proof}
Fix a column index $j \in \{1,\ldots,N\}$ and define $\mathcal{S}_0 \triangleq \{j\}$. Then, for $h=1,\ldots,N-1$, the following can be written as
\begin{align*}
    \mathcal{S}_h \triangleq \mathcal{S}_{h-1} \cup \bigcup_{v\in \mathcal{S}_{h-1}} \mathcal{N}_v^{\mathcal{G}_h},
\end{align*}
where $\mathcal{N}_v^{\mathcal{G}_h}$ denotes the neighborhood of $v$ in $\mathcal{G}_h$. We first show by induction on $h$ that
\begin{align*}
i \in \mathcal{S}_h \quad\Longrightarrow\quad \big[\mathbf{C}_h \mathbf{C}_{h-1}\cdots \mathbf{C}_1\big]_{ij} > 0.
\end{align*}

For $h=0$, the claim is immediate since $\mathcal{S}_0=\{j\}$. Assume it holds for some $h-1\ge 0$ and let $i\in \mathcal{S}_h$. If $i\in \mathcal{S}_{h-1}$, then $(\mathbf{C}_h)_{ii}=1$, and therefore
\begin{align*}
\big[\mathbf{C}_h\cdots\mathbf{C}_1\big]_{ij} \ge (\mathbf{C}_h)_{ii} \big[\mathbf{C}_{h-1}\cdots\mathbf{C}_1\big]_{ij} >0.
\end{align*}

Otherwise, $i\notin \mathcal{S}_{h-1}$ and, by construction of $\mathcal{S}_h$, there exists a vertex $v\in \mathcal{S}_{h-1}$ such that $(i,v)\in \mathcal{E}(\mathcal{G}_h)$. Hence, $(\mathbf{C}_h)_{iv}=1$, and by the induction hypothesis, $\big[\mathbf{C}_{h-1}\cdots\mathbf{C}_1\big]_{vj}>0$. Consequently,
\begin{align*}
\big[\mathbf{C}_h\cdots\mathbf{C}_1\big]_{ij} \ge (\mathbf{C}_h)_{iv} \big[\mathbf{C}_{h-1}\cdots\mathbf{C}_1\big]_{vj} >0.
\end{align*}

Next, suppose that $\mathcal{S}_{h-1}\neq \{1,\ldots,N\}$. Since $\mathcal{G}_h$ is connected, there must exist an edge between $\mathcal{S}_{h-1}$ and its complement. Therefore, at least one new vertex is added to $\mathcal{S}_h$, implying
\begin{align*}
|\mathcal{S}_h| \ge |\mathcal{S}_{h-1}|+1.
\end{align*}
As $|\mathcal{S}_0|=1$, it follows that
\begin{align*}
|\mathcal{S}_{N-1}| \ge 1+(N-1) = N.
\end{align*}
So, $\mathcal{S}_{N-1}=\{1,\ldots,N\}$, and thus $\big[\mathbf{C}_{N-1}\cdots\mathbf{C}_1\big]_{ij}>0$ for every $i$.

Because the choice of $j$ is arbitrary, every entry of $\mathbf{C}_{N-1}\cdots\mathbf{C}_1$ is strictly positive. Moreover, for fixed $N$, there are only finitely many connected simple graphs on $\{1,\ldots,N\}$, and hence only finitely many admissible sequences $\big(\mathcal{G}_1,\ldots,\mathcal{G}_{N-1}\big)$. Therefore,
\begin{align*}
\beta \triangleq \min_{i,j,\{\mathcal{G}_h\}} \big[\mathbf{C}_{N-1}\cdots\mathbf{C}_1\big]_{ij},
\end{align*}
is well defined and satisfies $\beta>0$. As a result, $\mathbf{C}_{N-1}\cdots\mathbf{C}_1 \ge \beta\,\mathbf{1}\mathbf{1}^\top$.
\end{proof}


\begin{theorem} \label{thm:convergence}
Consider the consensus protocol in Eq.~\eqref{consensus_law}. For an arbitrary
$\mathbf{X}_i(0) \in \mathbb{S}^n_{++}$,
\begin{align*}
    \lim_{\ell\to\infty} \mathbf{X}_i(\ell) = \mathbf{X}^\star, \quad i=1,\ldots,N,
\end{align*}
for some common $\mathbf{X}^\star \in \mathbb{S}^n_{++}$, with $\mathbf{X}^\star = \sum_{j=1}^N c_j \mathbf{X}_j(0)$, $c_j > 0$, $\sum_j c_j = 1$.
\end{theorem}

\begin{proof}
Since $\boldsymbol{\Phi}(\ell_2,\ell_1)$ is a product of consecutive matrices, it satisfies $\boldsymbol{\Phi}(\ell_3,\ell_1) = \boldsymbol{\Phi}(\ell_3,\ell_2) \, \boldsymbol{\Phi}(\ell_2,\ell_1)$ for any $\ell_1 \leq \ell_2 \leq \ell_3$. Hence, the product $\boldsymbol{\Phi}\big(\ell+(N-1)B,\ell\big)$ decomposes into $N-1$ consecutive length-$B$ blocks,
\begin{align}
    \boldsymbol{\Phi}\big(\ell+(N-1)B,\ell\big)
    &= \boldsymbol{\Phi}\big(\ell+(N-1)B,\ell+(N-2)B\big) \cdots \boldsymbol{\Phi}\big(\ell+B,\ell\big). \label{block-decomp}
\end{align}
Applying Lemma~\ref{lem:block-bound} to each of the $N-1$ factors gives, for
$h=0,\ldots,N-2$,
\begin{align*}
    \boldsymbol{\Phi}\big(\ell+(h+1)B,\ell+hB\big) \ \ge\ \gamma\big(\mathbf{A}_{\mathrm{u}}^{(\ell+hB)}+\mathbf{I}_N\big) \ \ge\ 0.
\end{align*}

Since all matrices involved are nonnegative, if $\mathbf{0}\le \mathbf{Z}_h \le \mathbf{V}_h$ for $h=1,\ldots,p$, then $\mathbf{Z}_1\cdots\mathbf{Z}_p \le \mathbf{V}_1\cdots\mathbf{V}_p$. Each entry of $\mathbf{V}_1\cdots\mathbf{V}_p$ is a sum of nonnegative terms, and each such term bounds the corresponding term of $\mathbf{Z}_1\cdots\mathbf{Z}_p$ from above. Applying this bound to each of the $N-1$ factors in Eq.~\eqref{block-decomp} and chaining them together gives, for every $\ell\ge0$,
\begin{align}
    \boldsymbol{\Phi} \big(\ell+(N-1)B,\ell\big) \ &\ge\ \gamma^{N-1} \big(\mathbf{A}_{\mathrm{u}}^{(\ell+(N-2)B)}+\mathbf{I}_N\big) \cdots \big(\mathbf{A}_{\mathrm{u}}^{(\ell)}+\mathbf{I}_N\big). \label{block-product}
\end{align}

By Assumption~\ref{ass:jc}, each union graph described as
\begin{align*}
\bigcup_{t=\ell+hB}^{\ell+(h+1)B-1}\mathcal{G}(t), \quad h=0,\ldots,N-2,
\end{align*}
is connected. Applying Lemma~\ref{lem:reachability} to the sequence of union graphs leads to
\begin{align}
\boldsymbol{\Phi} \big(\ell+(N-1)B,\ell\big) \ge \delta\,\mathbf{1}\mathbf{1}^{\top}, \quad \delta \triangleq \gamma^{N-1}\beta >0 .
\label{scrambling-bound}
\end{align}

Let $\tau(\mathbf{D})$ denote the Dobrushin ergodicity coefficient of a row-stochastic matrix $\mathbf{D}$, which is defined as
\begin{align*}
\tau(\mathbf{D}) = 1-\min_{i,r} \sum_{j=1}^{N} \min \left([\mathbf{D}]_{ij},[\mathbf{D}]_{rj} \right).
\end{align*}
The Dobrushin ergodicity coefficient satisfies $\tau(\mathbf{D}_1\mathbf{D}_2) \le\tau(\mathbf{D}_1)\tau(\mathbf{D}_2)$ and $\operatorname{osc}(\mathbf{D}\mathbf{d}) \le \tau(\mathbf{D})\operatorname{osc}(\mathbf{d})$ \cite{seneta1981non}. Since every entry of $\boldsymbol{\Phi} \big(\ell+(N-1)B,\ell \big)$ is bounded below by $\delta$ from Eq.~\eqref{scrambling-bound}, we get
\begin{align*}
\tau \left(\boldsymbol{\Phi} \big(\ell+(N-1)B,\ell\big)\right) \le 1-\delta .
\end{align*}

Consider $\mathbf{\Psi}_g \triangleq \boldsymbol{\Phi} \big( (g+1)(N-1)B,\, g(N-1)B \big)$. Using the submultiplicative property of $\tau$, it follows that
\begin{align}
\tau \left(\boldsymbol{\Phi}\big(G(N-1)B,0\big)\right) &= \tau \left(\mathbf{\Psi}_{G-1}\cdots\mathbf{\Psi}_0\right) \nonumber\\
&\le \prod_{g=0}^{G-1}\tau \left(\mathbf{\Psi}_g\right) \nonumber\\
&\le (1-\delta)^G \rightarrow 0.
\label{tau-convergence}
\end{align}

Equation~\eqref{tau-convergence} establishes that $\tau\big(\boldsymbol{\Phi}(\ell,0)\big) \to 0$ along the subsequence $\ell = G(N-1)B$ only. Extending the bound to arbitrary $\ell$ requires only that $\tau\big(\boldsymbol{\Phi}(\ell,0)\big)$ be non-increasing in $\ell$. Since $\boldsymbol{\Phi}(\ell+1,0) = \mathbf{W}(\ell) \, \boldsymbol{\Phi}(\ell,0)$ and $\tau(\mathbf{D})\le 1$ for every row-stochastic matrix $\mathbf{D}$, submultiplicativity of $\tau$ provides
\begin{align*}
    \tau\big(\boldsymbol{\Phi}(\ell+1,0)\big) &= \tau\big(\mathbf{W}(\ell)\boldsymbol{\Phi}(\ell,0)\big) \le \tau\big(\mathbf{W}(\ell)\big)\,\tau\big(\boldsymbol{\Phi}(\ell,0)\big) \le \tau\big(\boldsymbol{\Phi}(\ell,0)\big).
\end{align*}
Consequently, for any $\ell\ge0$, letting $G_\ell \triangleq \lfloor \ell/((N-1)B)\rfloor$
so that $G_\ell(N-1)B \le \ell$, we obtain
\begin{align}
    \tau\big(\boldsymbol{\Phi}(\ell,0)\big) \ \le\ \tau\big(\boldsymbol{\Phi}(G_\ell(N-1)B,0)\big) \ \le\ (1-\delta)^{G_\ell}, \label{tau-full}
\end{align}
and since $G_\ell \to \infty$ as $\ell \to \infty$, Eq.~\eqref{tau-full} gives $\tau\big(\boldsymbol{\Phi}(\ell,0)\big) \to 0$ as $\ell \to \infty$.

Fix a column index $j$ and consider the following
\begin{align*}
a_j(\ell) = \max_i \left[\boldsymbol{\Phi}(\ell,0)\right]_{ij}, \quad b_j(\ell) = \min_i \left[\boldsymbol{\Phi}(\ell,0)\right]_{ij}.
\end{align*}
Because each row of $\boldsymbol{\Phi}(\ell+1,0)$ is a convex combination of the rows of $\boldsymbol{\Phi}(\ell,0)$, the sequence $a_j(\ell)$ is non-increasing, while $b_j(\ell)$ is non-decreasing. Taking $\mathbf{d}=\mathbf{e}_j$ in the oscillation inequality, where $\mathbf{e}_j \in \mathbb{R}^N$ denotes the $j$-th canonical basis vector, we get
\begin{align*}
\operatorname{osc} \big(\boldsymbol{\Phi}(\ell,0) \, \mathbf{e}_j\big) \le \tau\big(\boldsymbol{\Phi}(\ell,0)\big) \operatorname{osc}(\mathbf{e}_j) = \tau\big(\boldsymbol{\Phi}(\ell,0)\big).
\end{align*}
As $\boldsymbol{\Phi}(\ell,0) \, \mathbf{e}_j$ is the $j$-th column of $\boldsymbol{\Phi}(\ell,0)$, it follows that
\begin{align*}
\operatorname{osc} \big(\boldsymbol{\Phi}(\ell,0) \, \mathbf{e}_j\big) = a_j(\ell)-b_j(\ell),
\end{align*}
and therefore
\begin{align*}
a_j(\ell)-b_j(\ell) \le \tau \big(\boldsymbol{\Phi}(\ell,0)\big).
\end{align*}

Combining this with Eq.~\eqref{tau-full} shows that $a_j(\ell)-b_j(\ell)\rightarrow 0$. Since both sequences are bounded and monotone, there exists $c_j\in[0,1]$ such that $a_j(\ell)\rightarrow c_j$ and $b_j(\ell)\rightarrow c_j$. Hence,
\begin{align*}
\left[\boldsymbol{\Phi}(\ell,0)\right]_{ij} \rightarrow c_j, \quad \forall\, i,j,
\end{align*}
which implies that
\begin{align*}
\boldsymbol{\Phi}(\ell,0) \rightarrow \mathbf{1}\mathbf{c}^{\top},
\end{align*}
where $\mathbf{c}^{\top}$ denotes a row vector. Since every $\boldsymbol{\Phi}(\ell,0)$ is row-stochastic, it follows that $c_j\ge 0$ and $\sum_{j=1}^{N} c_j = 1$. Moreover, from Eq.~\eqref{scrambling-bound} with $\ell=0$, every entry of $\boldsymbol{\Phi}\big((N-1)B,0\big)$ is at least $\delta$. Thus, $b_j\big((N-1)B\big)\ge \delta$, and as $b_j(\ell)$ is non-decreasing, the following can be represented by
\begin{align*}
c_j = \lim_{\ell\to\infty} b_j(\ell) \ge \delta >0, \quad j=1,\ldots,N.
\end{align*}

Finally, from Eq.~\eqref{collective2}, $\mathbf{s}(\ell) = \big( \boldsymbol{\Phi}(\ell,0) \otimes \mathbf{I}_{n^2} \big) \, \mathbf{s}(0)$, and hence the following can be written as
\begin{align*}
\lim_{\ell\to\infty} \mathbf{s}(\ell) = \big( \mathbf{1}\mathbf{c}^{\top} \otimes \mathbf{I}_{n^2} \big) \, \mathbf{s}(0).
\end{align*}
Equivalently, the limit of each information matrix can be expressed as
\begin{align*}
\lim_{\ell\to\infty} \mathbf{X}_i(\ell) = \sum_{j=1}^{N} c_j\,\mathbf{X}_j(0) \triangleq \mathbf{X}^{\star}, \quad i=1,\ldots,N.
\end{align*}

Since $\mathbf{X}_j(0)\in\mathbb{S}_{++}^{n}$ for all $j$ and $c_j>0$ with $\sum_{j=1}^{N} c_j=1$, the matrix $\mathbf{X}^{\star}$ is a convex combination of positive definite matrices and therefore $\mathbf{X}^{\star}\in\mathbb{S}_{++}^{n}$.
\end{proof}


\begin{corollary}
Under the conditions of Theorem~2, for every $i = 1, \dots, N$ and every $\ell \geq 0$,
\begin{align*}
\|\mathbf{X}_i(\ell) - \mathbf{X}^\star\|_F \; \leq \; (1-\delta)^{G_\ell} \sum_{j=1}^{N} \|\mathbf{X}_j(0)\|_F,
\end{align*}
where $G_\ell$ and $\delta = \gamma^{N-1}\beta > 0$ are as defined in the proof of Theorem~2. Consequently, $\mathbf{X}_i(\ell)$ converges to $\mathbf{X}^\star$ at a geometric rate.
\end{corollary}

\begin{proof}
From Eq.~\eqref{collective2}, $\operatorname{vec}\big(\mathbf{X}_i(\ell)\big) - \operatorname{vec}\big(\mathbf{X}^\star\big) = \sum_{j=1}^N \left(\big[\boldsymbol{\Phi}(\ell,0)\big]_{ij} - c_j\right)\operatorname{vec}\big(\mathbf{X}_j(0)\big)$. Since $a_j(\ell)$ is non-increasing to $c_j$ and $b_j(\ell)$ is non-decreasing to $c_j$, with $b_j(\ell) \le \big[\boldsymbol{\Phi}(\ell,0)\big]_{ij} \le a_j(\ell)$ for every $i$, it follows that $|[\boldsymbol{\Phi}(\ell,0)]_{ij} - c_j| \le a_j(\ell) - b_j(\ell) \le \tau\big(\boldsymbol{\Phi}(\ell,0)\big)$. Applying the triangle inequality and Eq.~\eqref{tau-full} then gives
\begin{align*}
\|\mathbf{X}_i(\ell) - \mathbf{X}^\star\|_F &\; \le \; \tau\big(\boldsymbol{\Phi}(\ell,0)\big) \sum_{j=1}^N \|\mathbf{X}_j(0)\|_F\\
&\; \le \;(1-\delta)^{G_\ell} \sum_{j=1}^N \|\mathbf{X}_j(0)\|_F.
\end{align*}
Hence, $\mathbf{X}_i(\ell)$ converges to $\mathbf{X}^\star$ at a geometric rate.
\end{proof}

\begin{remark}
The bound in Corollary~1 is geometric in the iteration index $\ell$, with effective decay factor $(1-\delta)$ per block of $(N-1)B$ iterations. Since $\delta = \gamma^{N-1}\beta$ with $\gamma = \alpha(\alpha^2/2)^{B-1}$, the guaranteed rate degrades as either the joint connectivity window $B$ or the network size $N$ increases.
\end{remark}

\begin{remark}
    Theorem~\ref{thm:convergence} addresses only the agreement among agents and does not establish whether the fused estimates remain consistent, which must be addressed separately.
\end{remark}


\subsection{Consistency Analysis}

The following result shows that the consensus step in Eq.~\eqref{consensus_law} preserves consistency at every iteration. Let $\mathbf{X}_i(\ell) \boldsymbol{\zeta}_i(\ell)$ denote the information vector associated with the information matrix $\mathbf{X}_i(\ell)$, where $\boldsymbol{\zeta}_i(0) = \bar{\mathbf{x}}_i$ and $\widetilde{\mathbf{x}}_i = \boldsymbol{\zeta}_i(M)$.

\begin{theorem} \label{thm:consistency}
Suppose that the pair $\big(\boldsymbol{\zeta}_i(\ell), \mathbf{X}_i(\ell)^{-1}\big)$ is consistent, i.e., $\mathbf{X}_i(\ell)^{-1} \succeq \acute{\mathbf{P}}_i(\ell)$, for all $i=1,\ldots,N$, where
\begin{align*}
    \acute{\mathbf{P}}_i(\ell) \triangleq \mathbb{E}\left[\big(\mathbf{x}-\boldsymbol{\zeta}_i(\ell)\big)\big(\mathbf{x}-\boldsymbol{\zeta}_i(\ell)\big)^\top\right],
\end{align*}
is the true error covariance of $\boldsymbol{\zeta}_i(\ell)$. The cross-covariances $\mathbb{E}\left[\big(\mathbf{x}-\boldsymbol{\zeta}_i(\ell)\big) \big(\mathbf{x}-\boldsymbol{\zeta}_j(\ell)\big)^\top\right]$, $i\ne j$, are not assumed to be known. Then, for every $i=1,\ldots,N$ and weights $w_{ij}(\ell)$ as defined in Eq.~\eqref{weights}, the pair $\big(\boldsymbol{\zeta}_i(\ell+1), \mathbf{X}_i(\ell+1)^{-1}\big)$ is also consistent.
\end{theorem}

\begin{proof}
Recall the update for $\mathbf{X}_i$ from Eq.~\eqref{consensus_law}. The corresponding update for the information vector is given by
\begin{align}
\mathbf{X}_i(\ell+1) \boldsymbol{\zeta}_i(\ell+1) = \sum_{j\in\mathcal{J}_i(\ell)} w_{ij}(\ell)\mathbf{X}_j(\ell)\boldsymbol{\zeta}_j(\ell),
\end{align}
where $w_{ij}(\ell)\ge 0$ and
\begin{align}
\sum_{j\in\mathcal{J}_i(\ell)} w_{ij}(\ell)=1.
\end{align}

These equations have the same form as the Covariance Intersection (CI) fusion rule expressed in information form \cite{julier1997non}. By assumption, each pair $\big(\boldsymbol{\zeta}_j(\ell),\mathbf{X}_j(\ell)^{-1}\big)$ is consistent, i.e.,
\begin{align*}
\mathbf{X}_j(\ell)^{-1} \succeq \acute{\mathbf{P}}_j(\ell), \quad j\in\mathcal{J}_i(\ell).
\end{align*}

The consistency result of CI \cite{julier1997non} states that fusing consistent information pairs using nonnegative normalized weights preserves consistency, irrespective of the unknown cross-correlations among the estimation errors. Therefore, the following can be written as
\begin{align*}
\mathbf{X}_i(\ell+1)^{-1} \succeq \acute{\mathbf{P}}_i(\ell+1),
\end{align*}
which implies that the pair $\big(\boldsymbol{\zeta}_i(\ell+1),\mathbf{X}_i(\ell+1)^{-1}\big)$ is also consistent.
\end{proof}

\begin{remark}
    Since $\mathbf{X}_i(0)^{-1} = \mathbf{P}_i$ is consistent by the consistency property of the local filter, applying Theorem~\ref{thm:consistency} inductively over $\ell=0,1,\ldots,M-1$ shows that every $\big(\boldsymbol{\zeta}_i(\ell),\mathbf{X}_i(\ell)^{-1}\big)$ remains consistent, and in particular so does $\big(\widetilde{\mathbf{x}}_i, \widetilde{\mathbf{P}}_i\big) = \big(\boldsymbol{\zeta}_i(M), \mathbf{X}_i(M)^{-1}\big)$.
\end{remark}


\section{Simulation Results} \label{sec:sim}

In this section, the performance of the proposed Information-Weighted Distributed Kalman Filter (IWDKF) is illustrated and compared with the centralized Kalman filter (CKF) and existing distributed filters: Kalman consensus filter (KCF) \cite{olfati2009kalman}, generalized Kalman consensus filter (GKCF) \cite{kamal2011generalized}, and information-weighted consensus filter (ICF) \cite{kamal2013information}. A multi-target tracking problem is considered in which a network of $N = 10$ agents monitors an area containing $N_T = 3$ ground targets that move within a bounded region. The FOV of each agent is randomly positioned such that the monitored region is fully covered collectively by the sensor network. The target dynamics are modeled by Eq.~\eqref{sys1} with the following parameters
\begin{align}
    \mathbf{F} = \begin{bmatrix}
    1 & 0 & T & 0\\
    0 & 1 & 0 & T\\
    0 & 0 & 1 & 0\\
    0 & 0 & 0 & 1
    \end{bmatrix},\;\; \mathbf{Q} = 5\begin{bmatrix}
    \frac{T^3}{3} & 0 & \frac{T^2}{2} & 0\\
    0 & \frac{T^3}{3} & 0 & \frac{T^2}{2}\\
    \frac{T^2}{2} & 0 & T & 0\\
    0 & \frac{T^2}{2} & 0 & T
    \end{bmatrix},
\end{align}
and $\mathbf{G} = \mathbf{I}_4$, where $T = 0.01$ is the sampling interval. The state vector $\mathbf{x} \in \mathbb{R}^4$ comprises positions in the $x$ and $y$ directions and their corresponding velocities. When agent $i$ observes target $t$, its measurement model is given by Eq.~\eqref{meas1} with the following parameters
\begin{align}
    \mathbf{H}^t_i = \begin{bmatrix}
        1 & 0 & 0 & 0\\
        0 & 1 & 0 & 0
    \end{bmatrix},\;\; \mathbf{R}_i^t = \begin{bmatrix}
        25 & 0\\
        0 & 25
    \end{bmatrix}.
\end{align}
When target $t$ is outside the sensing range of agent $i$, the corresponding matrices are set to $\mathbf{H}^t_i(k) = \mathbf{0}$ and $(\mathbf{R}_i^t)^{-1} = \mathbf{0}$. The CKF with access to all sensor measurements is used as the performance benchmark. For all distributed filters, each agent communicates with its neighbors only once per time step.

Each agent independently initializes its local estimate of each target's state. The position and velocity components are uniformly sampled from $[0,100]$ and $[-5,5]$, respectively. The CKF is initialized using the average of the initial state estimates across all agents.


\subsection{Fixed Network Topology}

As a preliminary comparison, the proposed IWDKF is evaluated against existing distributed filters under a fixed communication topology, as illustrated in Fig.~\ref{fig:graph}. In this setting, the distributed filters operate with their optimal consensus step sizes throughout the simulation, representing a favorable scenario for these algorithms. This baseline case serves as a reference for the more practically relevant case of switching network topologies considered in the next subsection.

This is a sparse network with $\Delta = 2$. The consensus step size is chosen as $\epsilon = 0.65/\Delta$ for all applicable algorithms. It is assumed that only agent 2 can sense the targets directly, i.e., $\mathcal{D}^t(k) = \{2\}$ and $\mathcal{O}^t(k) = \{4,5,\ldots,10\}$, for all $k \in \mathbb{Z}^+$.

\begin{figure}
    \centering
    \includegraphics[width=0.30\textwidth]{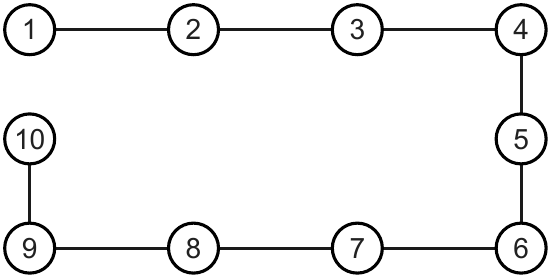}
    \caption{Communication network graph with $N=10$.}
    \label{fig:graph}
\end{figure}

Figure \ref{fig:prmse_fix} shows the position root mean square error (PRMSE) averaged over 10 agents and 3 targets based on 300 random simulation runs, for the five filters under consideration. Among the distributed filters, IWDKF achieves the lowest average PRMSE, second only to the CKF. The presence of multiple oblivious agents results in an initially high PRMSE for KCF, delaying its convergence.

Figure~\ref{fig:ae_fix} illustrates the absolute error (AE) in the posterior position estimates for each agent under IWDKF and CKF. The estimates of all IWDKF agents converge to the CKF estimate, with agents located farther from the observer requiring more time steps to converge.

These results indicate that IWDKF retains comparable estimation performance even when the competing distributed filters operate under favorable fixed-topology conditions.

\begin{figure}
    \centering
    \includegraphics[width=0.55\textwidth]{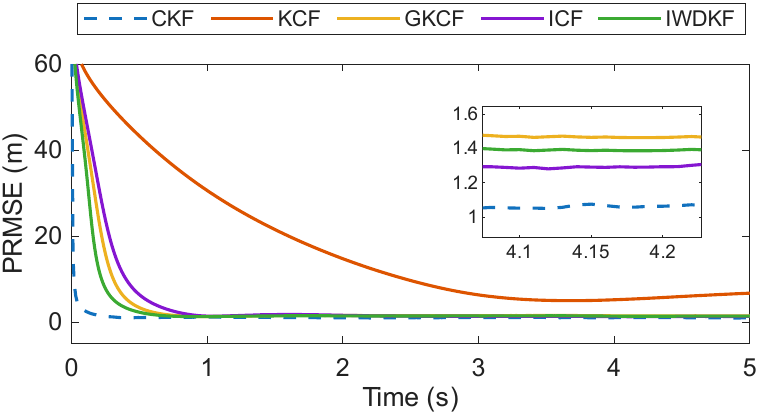}
    \caption{Position root mean square error (PRMSE) averaged over 10 agents and 3 targets based on 300 random simulation runs under a fixed network topology.}
    \label{fig:prmse_fix}
\end{figure}

\begin{figure}
    \centering
    \includegraphics[width=0.55\textwidth]{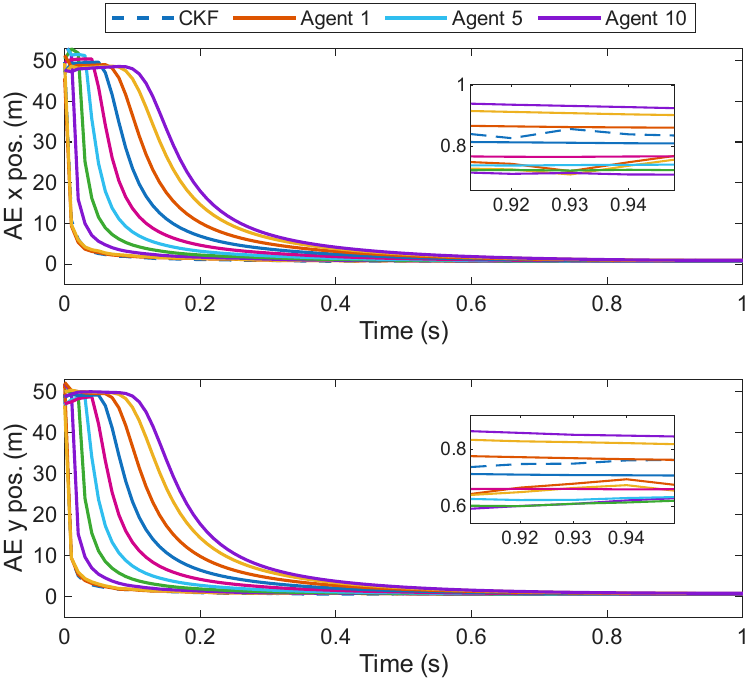}
    \caption{Absolute errors in posterior position estimates of IWDKF and CKF under a fixed network topology.}
    \label{fig:ae_fix}
\end{figure}


\subsection{Switching Network Topology}

The proposed IWDKF is now evaluated under a switching communication topology, where joint connectivity over a bounded time window is assumed, with $B = 10$. The network topology is initialized as the graph in Fig.~\ref{fig:graph} and subsequently evolves randomly over time. It is assumed that the observer set $\mathcal{D}^t(k)$ follows a fixed periodic schedule, namely $\mathcal{D}^t(k) \in \{ \{1\},\{2\},\ldots,\{10\}\}$. For all applicable algorithms, $\epsilon = 0.65/\Delta_0$, where $\Delta_0$ is the maximum degree of the initial topology. The results are obtained by averaging over 300 random simulation runs, each with a distinct switching sequence that satisfies the joint connectivity condition above. Figure~\ref{fig:graphstats} presents the number of active edges and the minimum, average, and maximum node degree over the simulation horizon for a single random simulation run.

\begin{figure}
    \centering
    \includegraphics[width=0.55\textwidth]{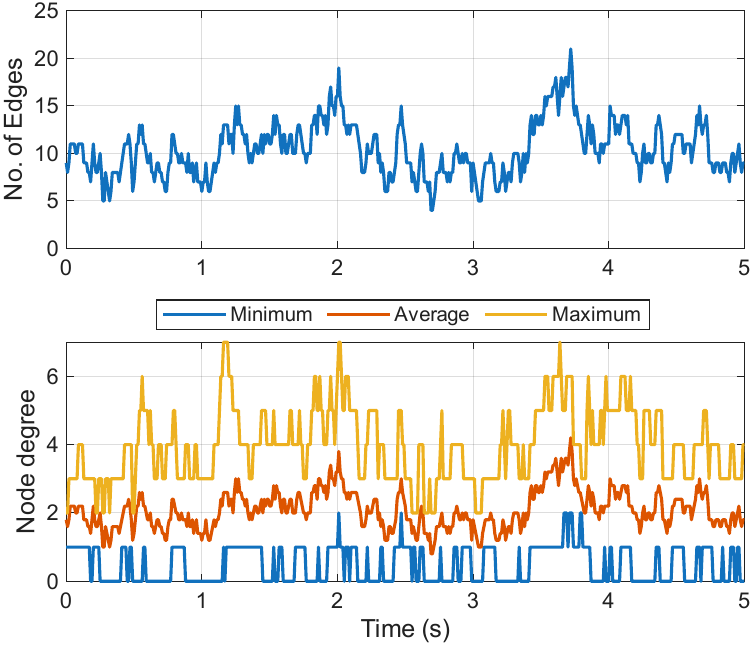}
    \caption{Evolution of graph metrics over the simulation horizon.}
    \label{fig:graphstats}
\end{figure}


Figure~\ref{fig:prmse_sw} shows the PRMSE averaged over 10 agents and 3 targets based on 300 random simulation runs, for the five filters under consideration. IWDKF achieves the lowest average PRMSE among the distributed filters, with performance close to that of the CKF. The intermittent jumps in the PRMSE of GKCF and ICF are due to the changing maximum degree of the graph. This results in the selected $\epsilon$ exceeding its upper limit $1/\Delta$ at some time instants and leading to momentary divergence of estimates. KCF is unaffected because $\epsilon$ appears only in an additive consensus term whose gain is bounded, so an outdated $\epsilon$ can only cause a bounded deviation in the estimate.

\begin{figure}
    \centering
    \includegraphics[width=0.55\textwidth]{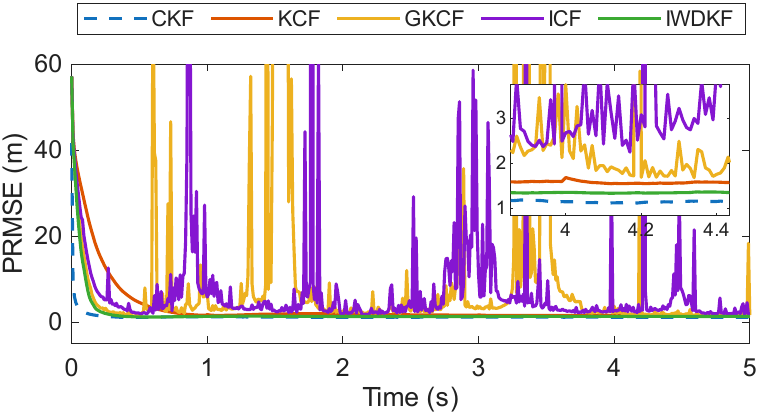}
    \caption{Position root mean square error (PRMSE) averaged over 10 agents and 3 targets based on 300 random simulation runs under a switching network topology.}
    \label{fig:prmse_sw}
\end{figure}


Figure~\ref{fig:post_err} depicts the posterior estimation error for the $x$-position of target 2 at agent 4, obtained by each filter. The $\pm 3\sigma$ bounds corresponding to each filter are also plotted as dashed lines using the same color scheme. The posterior estimation error of the proposed IWDKF remains within its $\pm 3\sigma$ bounds, whereas the posterior estimation error of ICF exceeds its $\pm 3\sigma$ bounds at multiple points.

\begin{figure}
    \centering
    \includegraphics[width=0.55\textwidth]{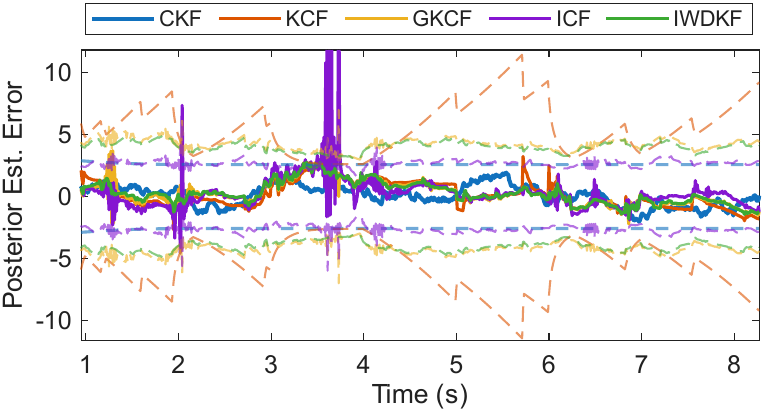}
    \caption{Posterior estimation error with $\pm 3\sigma$ bounds for the $x$-position of target 2 at agent 4 under a switching network topology.}
    \label{fig:post_err}
\end{figure}


Figure~\ref{fig:2d} illustrates the tracking behavior of two agents using the IWDKF against the CKF for target 3. Along each trajectory, markers are placed at four instants spread evenly across the latter half of the time horizon. At every marked instant, a dashed ellipse in the corresponding color depicts 95\% confidence region implied by that estimator's own covariance at that instant. The true target position falls within each ellipse at every checkpoint shown for both agents. This indicates that the covariances are not overconfident and that the IWDKF's local uncertainty estimates remain consistent.

\begin{figure}
    \centering
    \includegraphics[width=0.55\textwidth]{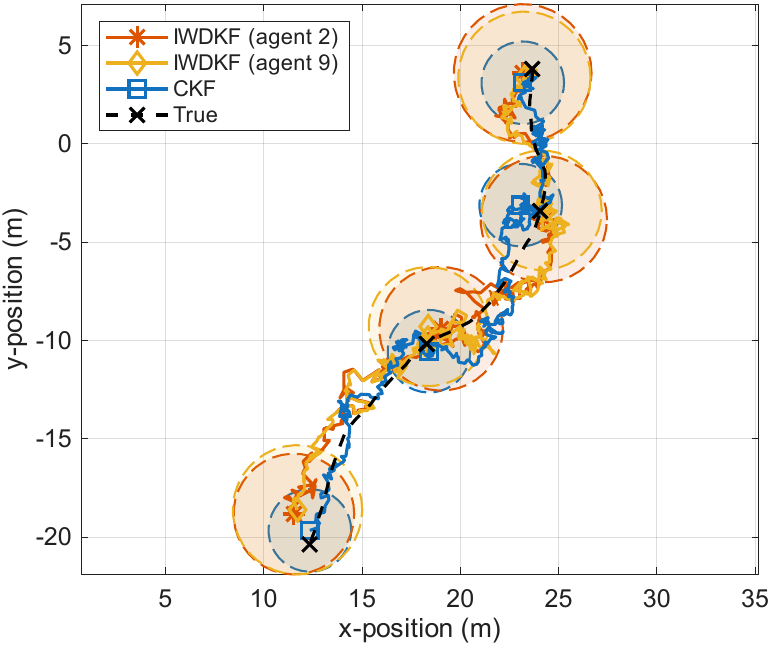}
    \caption{Trajectory of target 3 with CKF and IWDKF position estimates under a switching network topology.}
    \label{fig:2d}
\end{figure}


To assess filter consistency under the switching topology, Fig.~\ref{fig:nees} shows the normalized estimation error squared (NEES), averaged over all agents and targets, together with $95\%$ confidence bounds obtained from the $\chi^2$ distribution with $n$ degrees of freedom. The IWDKF estimates remain within the confidence bounds throughout the simulation, indicating that consistency is preserved even with the intermittent loss of connectivity.

\begin{figure}
    \centering
    \includegraphics[width=0.55\textwidth]{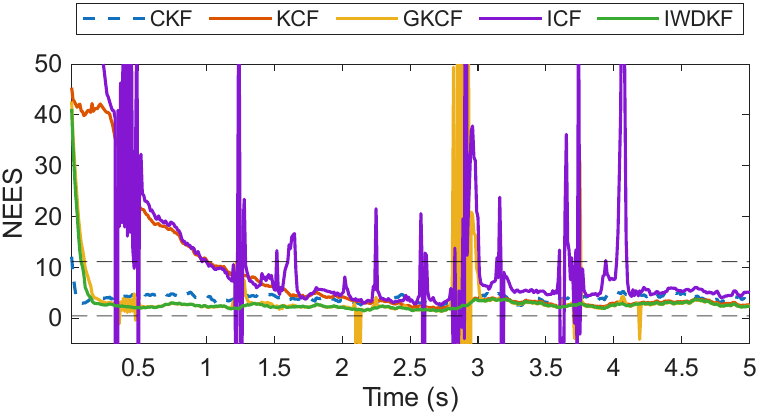}
    \caption{Normalized estimation error squared (NEES) with $95\%$ confidence bounds under a switching network topology.}
    \label{fig:nees}
\end{figure}


Figure~\ref{fig:convtime} shows the convergence time for each filter, defined as the earliest time instant at which the average PRMSE falls below a threshold of 5 m and remains below 5 m for at least 20 consecutive time steps. The IWDKF demonstrates faster convergence than the other distributed filters. Table~\ref{tab:metrics} presents a comparison of the performance metrics of the considered filters. These results demonstrate that the IWDKF retains both its estimation accuracy and its statistical consistency under jointly connected switching network topologies.


\begin{figure}
    \centering
    \includegraphics[width=0.55\textwidth]{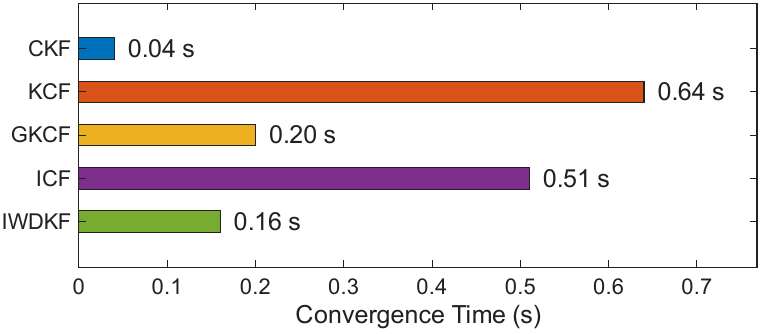}
    \caption{Convergence time comparison of the filters under a switching network topology.}
    \label{fig:convtime}
\end{figure}


\begin{table}
    \centering
    \caption{Comparison of performance metrics.}
    {\begin{tabular}{lcc} \toprule
    Name of filter & Avg. PRMSE (m) & Conv. time (s)\\ \midrule
    CKF & 1.2440 & 0.04\\
    KCF & 2.4518 & 0.64\\
    GKCF & 7.7225 & 0.20\\
    ICF & 9.1048 & 0.51\\
    IWDKF & 1.6648 & 0.16\\ \bottomrule
    \end{tabular}}
    \label{tab:metrics}
\end{table}


\subsection{Varying Numbers of Oblivious Agents}

A comparison of filter performance for varying numbers of oblivious agents is shown in Fig.~\ref{fig:oblivious}. The baseline graph in Fig.~\ref{fig:graph} is used, with a single observer agent $\mathcal{D}^t(k) = \{1\}$ and eight oblivious agents $\mathcal{O}^t(k) = \{3,4,\ldots,10\}$, for all $k \in \mathbb{Z}^+$. In each successive simulation, one additional edge is introduced between the observer and an oblivious agent, so that $|\mathcal{O}^t|$ decreases from 8 to 0. The IWDKF achieves a lower average PRMSE than ICF, GKCF, and KCF when multiple oblivious agents are present. The performance gap narrows with the decrease in the number of oblivious agents, since more agents then have access to direct or neighbor measurements.

\begin{figure}
    \centering
    \includegraphics[width=0.55\textwidth]{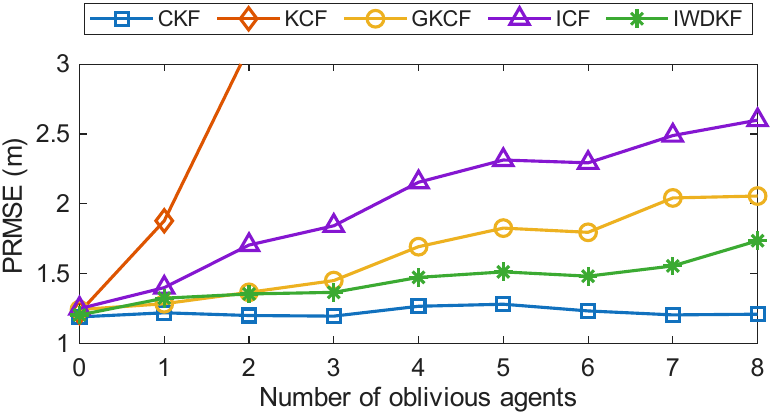}
    \caption{Comparison of position root mean square error (PRMSE) with varying number of oblivious agents.}
    \label{fig:oblivious}
\end{figure}


\subsection{Varying Numbers of Observer Agents}

Figure~\ref{fig:observer} depicts how filter performance varies with the number of observer agents. The base graph of the previous subsection is retained. In each successive simulation, one additional agent is added to the observer set, so that $|\mathcal{D}^t|$ increases from 1 to 10. With fewer observers, more agents must rely on information propagated through the network rather than on local measurements. As the number of observer agents decreases, the performance gap between IWDKF and the other distributed filters widens, even as the average PRMSE of all filters increases.

\begin{figure}
    \centering
    \includegraphics[width=0.55\textwidth]{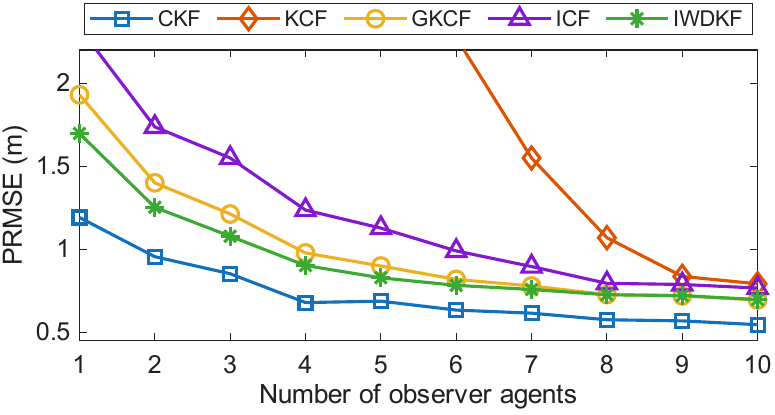}
    \caption{Comparison of position root mean square error (PRMSE) with varying number of observer agents.}
    \label{fig:observer}
\end{figure}


\section{Conclusion} \label{sec:conclusion}

This paper proposed an Information-Weighted Distributed Kalman Filter (IWDKF) for multi-target tracking using a network of sensor agents with limited sensing capabilities and switching communication topologies. The proposed approach relaxes several restrictive assumptions commonly adopted in existing distributed filtering methods, including the requirement for direct target observability, knowledge of global network parameters, and fixed communication topologies. The positive definiteness of the covariance matrix associated with the weighted-consensus estimate is established, and the consistency of the consensus estimate is shown. Furthermore, the proposed weighting strategy is shown to guarantee consensus at a geometric rate under jointly connected switching graphs. Simulation results for a multi-target tracking scenario with a network of agents demonstrated that the proposed IWDKF can accurately estimate target states despite limited local observations and dynamic network connectivity, while requiring only a single communication exchange per time step. The performance of IWDKF was also compared with that of the CKF and existing distributed filters: KCF, GKCF, and ICF, for both fixed and switching topologies. Future work can focus on extending the proposed IWDKF framework to nonlinear target dynamics and measurement models.


\bibliographystyle{ieeetr}
\bibliography{mybib}


\end{document}